\documentclass[orivec]{llncs}

\usepackage{listings}
\usepackage{graphicx}
\usepackage{subfig}
\usepackage{float}
\usepackage{hyperref}
\usepackage{wrapfig}
\usepackage{verbatim}

\usepackage{amsmath}
\usepackage{paralist}
\usepackage{relsize}
\usepackage{montserrat}
\usepackage{xspace}

\usepackage[dvipsnames]{xcolor}

\definecolor{lightblue}{RGB}{231,255,255}
\definecolor{lightred}{RGB}{255,224,224}
\definecolor{lightgreen}{RGB}{224,255,224}
\definecolor{lightyellow}{RGB}{255,255,224}
\definecolor{lightpurple}{RGB}{255,224,255}
\definecolor{darkerred}{RGB}{64,0,0}
\definecolor{darkred}{RGB}{128,0,0}
\definecolor{darkblue}{RGB}{0,0,128}
\definecolor{darkgreen}{RGB}{0,128,0}
\definecolor{darkpurple}{RGB}{128,0,128}
\definecolor{black}{RGB}{0,0,0}

\newcommand{\colorpar}[3]{}%{\colorbox{#1}{\parbox{#2}{#3}}}
\newcommand{\marginremark}[3]{}%{\marginpar{\colorpar{#2}{\linewidth}{\color{#1}#3}}}
\newcommand{\remarkREM}[1]{}%{\marginremark{darkblue}{white}{\scriptsize{[REM]~ #1}}}
\newcommand{\remarkPRD}[1]{}%{\marginremark{darkred}{lightred}{\tiny{[PRD]~ #1}}}
\newcommand{\remarkMDL}[1]{}%{\marginremark{darkyellow}{lightgreen}{\tiny{[MDL]~ #1}}}
\newcommand{\CEB}[1]{}%{\marginremark{PineGreen}{GreenYellow!20}{\tiny{\sffamily[CEB]~ #1}}}
\newcommand{\CEBwashere}{}%{\CEB{toqueti\'e ac\'a}}

\makeatletter
\def\THICKhrulefill{\leavevmode \leaders \hrule height 5pt\hfill \kern \z@}
\makeatother
\newcommand{\hrmkPRD}[1]{}%{\highlightedremark{darkred}{lightred}{PRD}{#1}}
\newcommand{\hrmkREM}[1]{}%{\highlightedremark{darkblue}{lightblue}{REM}{#1}}
\newcommand{\hrmkCEB}[1]{}%{\highlightedremark{PineGreen}{GreenYellow!20}{CEB}{#1}}

\newcommand{\fontrft}[1]{\ensuremath{\mathsf{#1}}}

\newcommand{\rft}{\text{RFT}}
\newcommand{\E}{\mathcal{E}}
\newcommand{\PAND}{\fontrft{pand}}
\newcommand{\AND}{\fontrft{and}}
\newcommand{\OR}{\fontrft{or}}
\newcommand{\VOT}{\fontrft{vot}}
\newcommand{\FDEP}{\fontrft{fdep}}
\newcommand{\SP}{\fontrft{sbe}}

\newcommand{\SG}{\fontrft{sg}}
\newcommand{\BE}{\fontrft{be}}
\newcommand{\RBOX}{\fontrft{rbox}}

\newcommand{\type}{\textrm{\textup{t}}}
\newcommand{\lab}{\textrm{\ensuremath{l}}}
\newcommand{\inputs}{\ensuremath{i}}
\newcommand{\spareinputs}{\ensuremath{si}}
\newcommand{\edges}{\text{E}}
\newcommand{\ver}{\ensuremath{V}}
\newcommand{\arity}{\mathcal{\#}}
\newcommand{\code}[1]{\text{\small\ttfamily #1}} % in text code
\newcommand{\tv}[1]{\textsf{\textit{\small #1}}} % text style variable
\newcommand{\cv}[1]{\textsf{\textit{\scriptsize #1}}} % code style variable

\newcommand{\fig}{{{\normalfont\textsmaller{\montserratalt{FIG}}}}\xspace}

\newcommand{\clocks}{\mathcal{C}}
\newcommand{\actions}{\mathcal{A}}
\newcommand{\inactions}{\actions^{\mathsf i}}
\newcommand{\outactions}{\actions^{\mathsf o}}
\newcommand{\coactions}{\actions^{\mathsf u}}
\newcommand{\states}{\mathcal{S}}
\newcommand{\trans}[1][]{\xrightarrow{\, {#1} \, }}

\newcommand{\dedrule}[2]{\frac{#1}{#2}}
\newcommand{\iosa}{IOSA}

\newcommand{\I}{\mathcal{I}}
\newcommand{\pll}{||}

\newcommand{\act}{\mathit{a}}
\newcommand{\actb}{\mathit{b}}

\newcommand{\activeck}{\mathit{active}}
\newcommand{\enablingck}{\mathit{enabling}}

\usepackage{paralist}
\usepackage{amssymb}
\usepackage{mathrsfs}
\usepackage{xifthen}

\newcommand{\borel}[1][]{%
\ifthenelse{\equal{#1}{}}{\mathscr{B}}{\mathscr{B}(#1)}%
}

\renewcommand{\emptyset}{\varnothing}

\renewcommand{\phi}{\varphi}

\newcommand{\R}{\ensuremath{\mathbb{R}}}

\newcommand{\N}{\ensuremath{\mathbb{N}}}
\newcommand{\pathWrite}{\mathit{Path}}

\newcommand{\anypath}[1]{  %cuidado! originalmente era \path
  \ifthenelse{\equal{#1}{}}
  {\ensuremath{\pathWrite}}
  {\ensuremath{\pathWrite(#1)}}
}
\newcommand{\finpath}[1]{
  \ifthenelse{\equal{#1}{}}
  {\ensuremath{\pathWrite^*}}
  {\ensuremath{\pathWrite^*(#1)}}
}
\newcommand{\omegapath}[1]{
  \ifthenelse{\equal{#1}{}}
  {\ensuremath{\pathWrite^\omega}}
  {\ensuremath{\pathWrite^\omega(#1)}}
}

\title{ A compositional semantics for Repairable Fault Trees with general
        distributions
        \thanks{Supported by SeCyT-UNC 05/BP12, 05/B497 and ERC grant 695614 (POWVER).}
		\thanks{Also by NWO project 15474 (\emph{SEQUOIA}) and EU project 102112 (\emph{SUCCESS}).}}
\author{Ra\'ul Monti\inst{1}\inst{2},
Pedro R. D'Argenio\inst{1}\inst{2}\inst{3},
Carlos E. Budde\inst{4}} 
\institute{
Universidad Nacional de C\'ordoba, FAMAF, C\'ordoba, Argentina
\and CONICET, C\'ordoba, Argentina
\and Saarland University, Department of Computer Science, Saarbr\"ucken, Germany
\and UTWENTE}

\begin{document}

%%%%%%%%%%%%%%%%%%%%%%%%%%%%%%%%%%%%%%%%%%%%%%%%%%%%%%%%%%%%%%%%%%%%%%%%%%%%%%%%
\maketitle
%%%%%%%%%%%%%%%%%%%%%%%%%%%%%%%%%%%%%%%%%%%%%%%%%%%%%%%%%%%%%%%%%%%%%%%%%%%%%%%%

\remarkREM{Authors in alphabetical order?}

%%%%%%%%%%%%%%%%%%%%%%%%%%%%%%%%%%%%%%%%%%%%%%%%%%%%%%%%%%%%%%%%%%%%%%%%%%%%%%%%
\begin{abstract}
Fault Tree Analysis (FTA) is a prominent technique in industrial and scientific
risk assessment. Repairable Fault Trees (RFT) enhance the classical Fault Tree
(FT) model by introducing the possibility to describe complex dependent repairs
of system components. Usual frameworks for analyzing FTs such as BDD, SBDD,
and Markov chains fail to assess the desired properties over RFT complex models,
\CEBwashere
either because these become too large, or due to cyclic behaviour introduced by
dependent repairs. Simulation is another way to carry
% either because of the size of the models or because of the cyclic behavior
% introduced by the dependent repairs. Simulation is another possible way to carry
out this kind of analysis. In this paper we review the RFT model with Repair
Boxes as introduced by Daniele Codetta-Raiteri. We present compositional
semantics for this model in terms of Input/Output Stochastic Automata, which
allows for the modelling of events occurring according to general continuous
distribution. Moreover, we prove that the semantics generates (weakly)
deterministic models, hence suitable for discrete event simulation, and
prominently for Rare Event Simulation using the \fig tool.
\hrmkCEB{Ponerme como autor con \texttt{\textbackslash{thanks}} para ``NWO project 15474 (\emph{SEQUOIA}) and EU project 102112 (\emph{SUCCESS})''}
\end{abstract}
%%%%%%%%%%%%%%%%%%%%%%%%%%%%%%%%%%%%%%%%%%%%%%%%%%%%%%%%%%%%%%%%%%%%%%%%%%%%%%%%

%%%%%%%%%%%%%%%%%%%%%%%%%%%%%%%%%%%%%%%%%%%%%%%%%%%%%%%%%%%%%%%%%%%%%%%%%%%%%%%%
\section{Introduction}\label{sec:introduction}
%%%%%%%%%%%%%%%%%%%%%%%%%%%%%%%%%%%%%%%%%%%%%%%%%%%%%%%%%%%%%%%%%%%%%%%%%%%%%%%%

Fault Tree Analysis is a prominent technique for dependability assessment of
complex industrial systems.
\CEBwashere
Standard or \emph{Static Fault Trees} (SFTs~\cite{haasl1981fault})
% from a graphical way of modelling.
% Standard or Static Fault Trees (SFT) \cite{haasl1981fault}
are DAGs whose leafs are called Basic Events (BE), and usually represent the
failure of a physical system component.
\CEBwashere
Each leaf is equipped with a failure rate or discrete probability, indicating
% Each leaf is equipped with a failure rate that indicates
the frequency at which
the component breaks. The other FT nodes are called gates, and they
model how basic components failures combine to induce more complex system
failures, until
\CEBwashere
the failure of interest (the \emph{top event} of the tree) occurs.
SFTs thus encode a logical formula. One of the most efficient analysis
techniques uses Binary Decision Diagrams (BDD) to represent the formula,
and then perform dependability studies using specialised algorithms.
This assumes the absence of stochastic dependency among BEs.
% eventually triggering the \emph{top event} of the tree, which represents the
% failure of interest.  The most efficient analysis technique on these trees
% consists in building a Binary Decision Diagram (BDD) representing the same
% formula as the FT and then solving the required dependability study by using
% several optimised algorithms, on the assumption of no stochastic dependency of
% the BEs.

Many extensions to SFTs allow for further modelling capabilities. One of the most
studied 
\CEBwashere
are \emph{Dynamic Fault Trees} (DFTs~\cite{159800,DBLP:conf/dsn/JungesGKS16}).
DFTs add gates to describe time- and order-dependence among the tree nodes,
% flavours is Dynamic Fault Trees (DFT), first introduced by Dugan
% \textit{et al.}~\cite{159800,DBLP:conf/dsn/JungesGKS16}.
% , which incorporates novel gates which have the particularity of introducing
% time and ordering dependence between the nodes they relate in the tree,
in contrast to the plain
combinatorial behavior of SFT gates. New analysis methods were introduced in
order to capture temporal requirements, such as cut sequences, translation
to Markov models~\cite{159800,159800,DBLP:conf/dsn/BoudaliCS07},
Sequence BDDs~\cite{DBLP:journals/ress/GeLYZC15,DBLP:journals/ress/Rauzy11,DBLP:journals/ress/XingSD11},
algebraic approaches~\cite{DBLP:journals/tr/MerleRLB10,amari2003new},
simulation, and combination and optimisations 
\CEBwashere
thereof~%
%of these methods~%
\cite{DBLP:journals/tse/BobbioFGP03,gulati1997modular}.

\emph{Repairable Fault Trees} (\rft~\cite{bobbio2004parametric,DBLP:conf/dsn/RaiteriIFV04,DBLP:conf/valuetools/BeccutiRFH08,DBLP:conf/dsn/BoudaliCS07})
increase FTs expressiveness by introducing the possibility to model complex
inter-dependent repair mechanisms for basic components,
\CEBwashere
i.e.\ system components that produce the basic events.
% (components of the system that produce the basic events.)
In former models such as DFT, certain
notions of repair had been addressed by allowing components to be repaired
independently. Nevertheless, this is not
\CEBwashere
usual in real world systems, where repair scheduling, resources management, and
maintenance play an important role.
% the usual case in real world systems, where matters such as repair scheduling,
% repair resources managing, and maintenance get involved.
%
To address this, we will focus on the \emph{Repair Box} model 
(RBOX~\cite{franceschinis2002towards,DBLP:conf/dsn/RaiteriIFV04}).
A RBOX models
a repair unit in charge of repairing certain BEs following certain
policy. Different repair policies such as \emph{first come first serve},
\emph{priority service}, \emph{random} or \emph{nondeterministic choice}, 
allow to analyze the impact of taking these decisions in the real
system. The introduction of these boxes greatly changes the dynamic of the tree.
Quantitative analyses are no longer a combinatorial calculation, since the
evolution of the system over time has to be
considered~\cite{DBLP:journals/csr/RuijtersS15}. Furthermore, traditional
qualitative analysis such as \emph{cut sets} lack of utility by not taking
repairability into account. Traditional quantitative analysis is also discarded
by the cyclic behavior introduced by this model which disallows to use
combinatorial solutions proposed for non repairable FTs and require a state
based solution instead~\cite{DBLP:journals/tse/BobbioFGP03}.

\hrmkCEB{Voy a sacar pedazos de estos párrafos y armar un bloque, que voy a presentar como la \texttt{\textbackslash{subsubsection}}\{\textbf{Related work}\}.}

In this work we present a formal definition of \emph{Repairable Fault Trees}
(\rft), along with its semantics given in terms of Input/Output Stochastic
Automata (\iosa) \cite{DBLP:conf/formats/DArgenioLM16,DArgenioMonti18}. We show
that the underlying \iosa\ semantics of the \rft\ specification is \emph{weakly
deterministic}, that is, the non-determinism present in the \iosa\ model is
spurious. Hence the model is equivalent to a fully stochastic model and thus
amenable to discrete event simulation. \iosa\ allows us to model \rft{s} general
continuous failure and repair distributions.

% %%%%%%%%%%%%%%%%%%%%%%%%%%%%%%%%%%%%%%%%%%%%%%%%%%%%%%%%%%%%%%%%%%%%%%%%%%%%%%%%
% \section{Related work}\label{sec:relatedwork}
% %%%%%%%%%%%%%%%%%%%%%%%%%%%%%%%%%%%%%%%%%%%%%%%%%%%%%%%%%%%%%%%%%%%%%%%%%%%%%%%%

\remarkREM{Related work can also start from here on... but maybe what Carlos
thinks to do is better...}

A variety of works address the problem of defining a rigorous syntax and
semantics to FT, DFT, and
\rft~\cite[etc.]{DBLP:conf/issre/CoppitSD00,bobbio2004parametric,DBLP:conf/dsn/BoudaliCS07,DBLP:conf/valuetools/BeccutiRFH08,DBLP:conf/dsn/BoudaliCS07}.
They usually differ, for example, in the types and meaning of gates,
expressiveness power, how spare elements are claimed and how repair races are
resolved. Presence of non-deterministic situations is also a main discording
issue. Comprehensive surveys on FTs can be found in~\cite{DBLP:conf/dsn/JungesGKS16}
and~\cite{DBLP:journals/csr/RuijtersS15}. In Section \ref{sec:rftsemantics} we
formally define the syntax for \rft{s} in a similar manner as
\cite{DBLP:conf/atva/BoudaliCS07} has done for DFTs. Furthermore, in order to
define the compositional and weakly deterministic semantics using \iosa, we
discuss different concerns about determinism on \rft{s}.

As discussed before, RFT analysis requires a state space solution. This usually
means one of the following two approaches. A first approach would be translating
the model to a Markov model, applying as much optimisations as possible during
the modelling and analysis in order to relieve the state explosion problem as
much as possible. This is the approach followed by many works such
as~\cite{DBLP:conf/valuetools/BeccutiRFH08,DBLP:journals/tse/BobbioFGP03,bobbio2004parametric}.
Two main drawbacks can be pointed out on this approach. The first one is that no
matter which existing optimisation methods are used, there is no guarantee that
there will be a significant state space reduction in general models. This is a
specially difficult situation in big and complex industrial size systems
analysis involving repair. A second drawback is the restriction to
exponentially distributed events, not allowing to correctly model real life
systems where timing is governed by other continuous distributions. This is the
case for example of phenomena such as timeouts in communication protocols, hard
deadlines in real-time systems, human response times or the variability of the
delay of sound and video frames (so-called jitter) in modern multi-media
communication systems, which are typically described by non-memoryless
distributions such as uniform, log-normal, or Weibull
distributions~\cite{DArgenioMonti18}.
A second approach to RFT analysis would be recurring to simulation, which does
not need the full state space of the model to be constructed, and does not
impose \emph{per se} the restriction to any kind of probabilistic distributions.
The main problem when confronting simulation is the big amount of computation
needed to reach a sufficiently accurate result. This is a most relevant issue
when analyzing highly dependable or fault tolerant systems, where the failure
probability is very small and plane Monte Carlo simulation becomes infeasible.
To face this problem one can make use of Rare Event Simulation techniques such
as Importance Splitting or Importance
Sampling~\cite{DBLP:reference/npe/Villen-AltamiranoV11,thesis/unc/Budde2017,DBLP:conf/epew/BuddeDH15,Rubino:2009:RES:1643623}.

Our main contribution in this work consists in a method for precisely modelling
RFTs with generally distributed events. Furthermore, by yielding a deterministic
\iosa\ model, thus amenable to discrete event simulation, we are able to
analyze it on the FIG Rare Event Simulation
Tool~\cite{DBLP:conf/epew/BuddeDH15,budde2016compositional}, greatly improving
efficiency when analyzing highly dependable systems. Also the recent work
\cite{DBLP:conf/safecomp/RuijtersRBS17} takes on the matter of using rare event
simulation to analyze DFTs with complex repairs. Nevertheless, they restrict to
Exponential and Erlang distributions and they finally conduce their analysis
over a Markov model hence suffering of potential states space explosion.

%===============================================================================
\section{Repair Fault Trees}\label{sec:RepairFaultTrees}
%===============================================================================

In Fig. \ref{fig:gates} we depict the set of \rft elements that we consider in
this work. Each of them has a set of inputs where to connect its subtrees, and
an output (if applicable) to propagate the failure, repair and other signals.
The propagation of a failure and its subsequent repair starts at the leafs of
the fault tree, including only (spare) basic elements. When one of them fails,
or gets repaired, it instantaneously propagates the event to the gates to which
it is connected. The state of a gate changes based on the signals it receives
from its inputs and propagates its new state to the gates it serves as input.
Thus, a proper combination and timing of fail signals may change a gate's state
to failing, and similarly, a proper combination and timing of repair signals may
change it back to a working state. This very much depends on the type of gate.
% At the same time, gates may output a signal when their state
% changes. Again it will be usually the case that they will output a fail signal
% when their state changes from working to failing, and a repair signal when it
% changes the other way round.
The state changes will at the same time trigger output signals accordingly. Not
only fail and repair signals, but also other signals may be produced, as it can
be in the case of repair boxes, which may output a start repairing signal to any
of their input basic elements.
\remarkREM{Esta explicación se puede abreviar.}

\begin{wrapfigure}[19]{l}{0.62\textwidth}
  \centering
  \includegraphics[scale=.18]{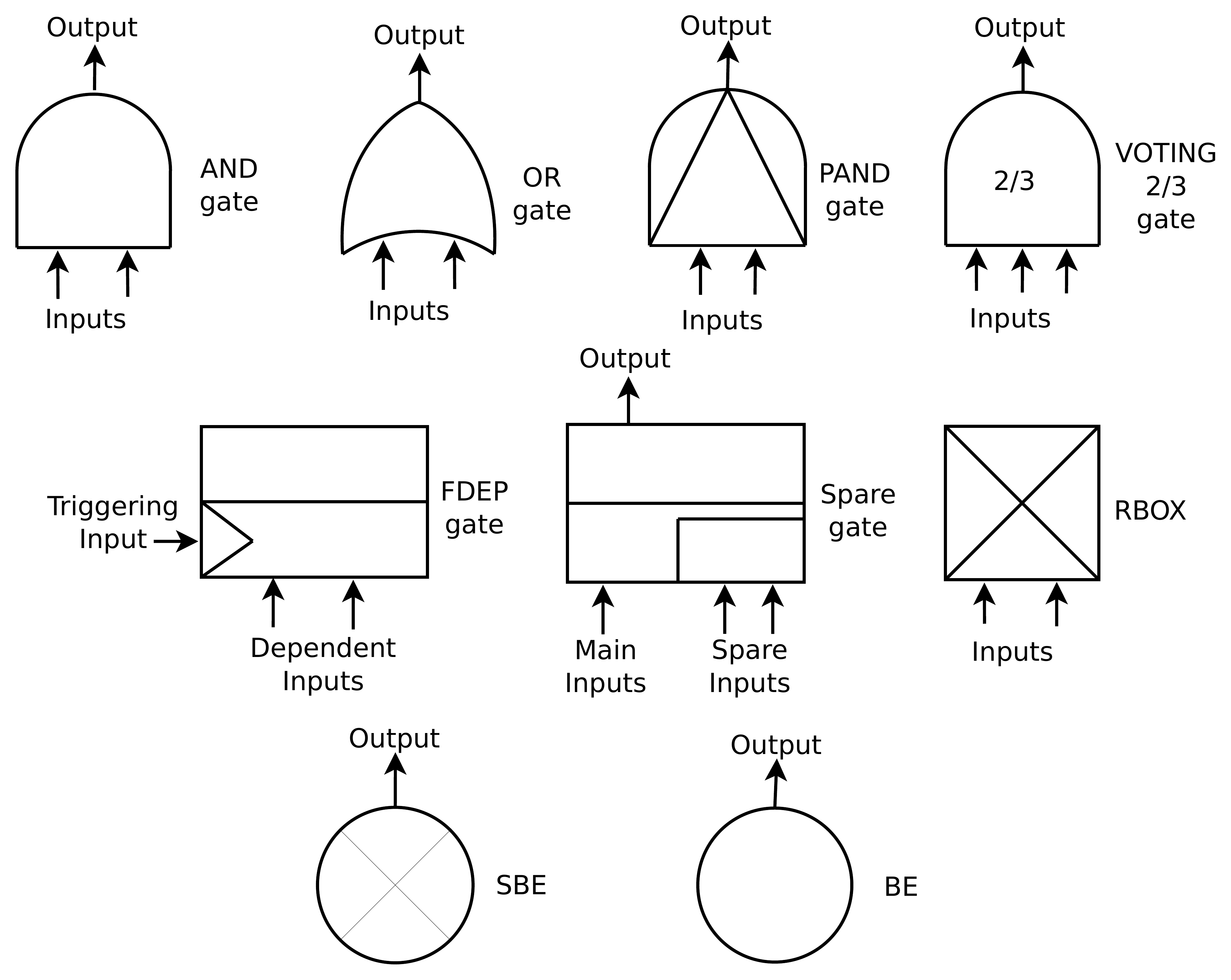}
  \caption{RFT elements}\label{fig:gates}
\end{wrapfigure}
The intuition about the behavior of each gate is as follows. An AND gate fails
whenever all its inputs fail, and gets repaired (stop failing) when at least one
of its inputs is repaired. An OR gate fails whenever at least one of its inputs
fails and is repaired when all of its inputs are repaired. A $k/n$ VOTING gate
fails whenever at least $k$ of its $n$ inputs fail and stops failing if at most
$k-1$ of its inputs remain failing. A PAND gate fails whenever its inputs fail
from left to right, inducing an order on the failure occurrence, and it is
repaired if the last input is repaired. A \emph{functional dependency gate}
(FDEP) has $n+1$ inputs. The fail signal of one of its inputs (the triggering
one) makes all the other inputs inaccessible to the rest of the system. Note
that the dependent inputs do not necessarily fail, and they will be accessible
again as soon as the triggering component is repaired (note the difference with
\cite{DBLP:conf/dsn/BoudaliCS07,DBLP:conf/safecomp/RuijtersRBS17} where
dependent BEs do fail). In fact this gate can be easily replaced by a system of
OR gates~\cite{xing2009efficient}. A \emph{spare basic element} (SBE) is a
special case of BE which can be enabled and disabled, and can be used as spare
parts for other BEs through spare gates. A same SBE can be shared by several
spare gates, and different sharing policies are introduced for this purpose. A
\remarkREM{Le explicacion de la SG es un poco redundante.}
\emph{spare gate} (SG) allows to replace a basic element by one of several spare
basic elements in case it fails. Each spare gate has a main input and $n$
spare parts inputs. The main input can only be a BE. The spare inputs can only
be SBEs. As soon as the main input fails, the SG uses its own policy to ask for
the replacement by one of its spare inputs. The SG will fail whenever it does
not obtain a replacement, and will signal repair whenever the main input gets
repaired or a spare input is obtained. If an in-use replacement fails the SG
will look for a new one. If the main input is repaired, the SG will free the
acquired spare input, in case there is one. A \emph{repair box} (RBOX) is the
unit in charge of managing the repairing of failed BEs and SBEs. They have $n$
inputs, which are the elements administered for repairing, and a dummy output. A
RBOX policy determines in which order the failing elements will be repaired.
Also notice that a RBOX can only repair one of its inputs at a time, while
the rest of its failing inputs are waiting for repair.

%===============================================================================
\section{Input/Output Stochastic Automata}\label{sec:iosa}
%===============================================================================

Input/Output Stochastic Automata
\cite{DBLP:conf/formats/DArgenioLM16,DArgenioMonti18} is a modelling formalism
tailored to model stochastic systems for the purpose of simulation. \iosa\
combine continuous probability jumps from Stochastic Automata, with discrete
event synchronisation for a compositional style of modelling. \iosa{s} use
continuous random variables to control and observe the passage of time. These
variables, called clocks, are set to a value according to their associated
probability distribution, and, as time evolves, count down all at the same rate
until they reach the value of zero. Clocks control the moments when actions are
taken, and thus allow to model systems where events occur at random continuous
time stamps.  Output and input transitions can be used to synchronize and
communicate between different \iosa{s}. Output transitions are autonomous, while
inputs occurrence depends on synchronisation with outputs. A transversal
classification for actions allows to mark them as urgent or non urgent.  While a
non-urgent output is controlled by the expiration of clocks (i.e., clocks
reaching the value zero), an urgent output action is taken as soon as the state
in which it is enabled is reached.  Though an \iosa\ may be non-deterministic,
\cite{DArgenioMonti18}~provides a set of sufficient conditions that guarantee
weak determinism (i.e. only spurious non-determinism is present).
Furthermore, such conditions can be checked with a polynomial algorithm on the
components of the model.
\begin{definition}\label{def:iosa}
  An \emph{input/output stochastic automaton with urgency} (\iosa) is a structure
  $(\states,\actions,\clocks,\trans,C_0,s_0)$, where $\states$
  is a (denumerable) set of states, $\actions$ is a (denumerable) set
  of labels partitioned into disjoint sets of \emph{input} labels
  $\inactions$ and \emph{output} labels $\outactions$, from which a subset
  $\coactions\subseteq\actions$ is marked as \emph{urgent},
  $\clocks$ is a (finite) set of clocks such that each $x\in\clocks$ has
  an associated continuous probability measure $\mu_x$ on $\R$
  s.t.\ $\mu_x(\R_{>0})=1$,
  ${\trans}\subseteq\states\times\clocks\times\actions
  \times\clocks\times S$ is a transition function, $C_0$ is the
  set of clocks that are initialized in the initial state, and $s_0\in
  \states$ is the initial state.
  In addition it should satisfy the following
  constraints:
  \begin{enumerate}[\rm a]
    \renewcommand{\theenumi}{(\alph{enumi})}
  \item\label{def:iosa:input-and-commit-are-reactive}%
    If $s\trans[C,\act,C']s'$ and $\act\in \inactions\cup\coactions$,
    then $C=\emptyset$.
  \item\label{def:iosa:output-is-generative}%
    If $s\trans[C,\act,C']s'$ and $\act\in \outactions\setminus\coactions$,
    then $C$ is a singleton set.
  \item\label{def:iosa:clock-control-one-output}%
    If $s\trans[\lbrace x\rbrace,\act_1,C_1]s_1$ and
    $s\trans[\lbrace x\rbrace,\act_2,C_2]s_2$
    then $\act_1=\act_2$, $C_1=C_2$ and $s_1=s_2$.
  \item\label{def:iosa:input-enabled}%
    For every $\act\in\inactions$ and state $s$, there exists a
    transition $s\trans[\emptyset,\act,C]s'$.
  \item\label{def:iosa:input-deterministic}%
    For every $\act\in\inactions$, if $s\trans[\emptyset,\act,C_1']s_1$ and
    $s\trans[\emptyset,\act,C_2']s_2$, $C_1' = C_2'$ and $s_1 = s_2$.
  \item\label{def:iosa:clock-never-go-off-early}%
    There exists a function $\activeck:\states\rightarrow 2^{\clocks}$ such that:
%    \begin{enumerate}[(i)]
    \begin{inparaenum}[i]
    \renewcommand{\theenumii}{\normalfont(\roman{enumii})}
    \item\label{def:iosa:clock-never-go-off-early:i}%
      $\activeck(s_0)\subseteq C_0$,
    \item\label{def:iosa:clock-never-go-off-early:ii}%
      $\enablingck(s)\subseteq \activeck(s)$,
    \item\label{def:iosa:clock-never-go-off-early:iii}%
      if $s$ is stable, $\activeck(s)=\enablingck(s)$, and
    \item\label{def:iosa:clock-never-go-off-early:iv}%
      if $t\trans[C,a,C']s$ then
            $\activeck(s)\subseteq(\activeck(t)\setminus C)\cup C'$.
    \end{inparaenum}
%    \end{enumerate}
  \end{enumerate}
  where $\enablingck(s)=\{y\mid s\trans[\{y\},\_,\_]\_\}$, and $s$ is
  \emph{stable} if there is no $a\in\coactions\cap\outactions$ such
  that $s\trans[\emptyset,a,\_]\_$. ($\_$ indicates the existential
  quantification of a parameter.)
\end{definition}

Restrictions (a) to (f) are there to ensure that at most one non-urgent output
action is enabled at a time. If in addition the \iosa\ is closed (i.e., all
communications have been resolved and hence the set of inputs is empty) and all
its urgent actions are confluent (in the sense
of~\cite{DBLP:books/daglib/0067019}, see also Def.~\ref{def:confluence}) it
turns out all the non-determinism is spurious and does not alter the stochastic
behavior (i.e. regardless of how non-determinism is resolved, the stochastic
properties remain the
same)~\cite{DBLP:conf/formats/DArgenioLM16,DArgenioMonti18}.  We call this
property \emph{weak determinism}.

\iosa{s} are closed under parallel composition which is defined
according to rules in Table~\ref{tb:parcomp}.  In order to avoid
unintended behavior, the component \iosa{s} are requested to be
\emph{compatible}, that is, they should not share output actions nor
clocks, and be consistent with respect to urgent actions.
\begin{table}[t]
  \vspace{-2em}
  \caption{Parallel composition on \iosa}\label{tb:parcomp}
%  \vspace{-1em}
  \begin{minipage}{.49\textwidth}
  \begin{gather}
    \dedrule{s_1\trans[C,\act,C']_1s_1'}{s_1\pll s_2\trans[C,\act,C']s_1'\pll s_2}%
    \ \act\in \actions_1{\setminus}\actions_2 \label{eq:par1}
  \end{gather}
  \end{minipage}
  \hfill
  \begin{minipage}{.49\textwidth}
  \begin{gather}
    \dedrule{s_2\trans[C,\act,C']_2s_2'}{s_1\pll s_2\trans[C,\act,C']s_1\pll s_2'}%
    \ \act\in \actions_2{\setminus}\actions_1 \label{eq:par2}
  \end{gather}
  \end{minipage}

  ~\hfill
  \begin{minipage}{.50\textwidth}
  \begin{gather}
    \dedrule{s_1\trans[C_1,\act,C'_1]_1s_1' \quad s_2\trans[C_2,\act,C'_2]_2s_2'}%
            {s_1\pll s_2\trans[C_1\cup C_2,\act,C'_1\cup C'_2]s_1'\pll s_2'}%
            \ \act\in \actions_1{\cap}\actions_2  \label{eq:par3}
  \end{gather}
  \end{minipage}
  \hfill~
\end{table}

%===============================================================================
\section{\iosa\ symbolic language}\label{sec:symbolicLanguage}
%===============================================================================

\begin{wrapfigure}[]{l}{0.59\textwidth}
\vspace{-2em}
\begin{lstlisting}
module BE
 fc, rc : clock;
 inform : [0..2] init 0;
 broken : [0..2] init 0;

 [fl!]     broken=0 @ fc -> (inform'=1) & (broken'=1);
 [r??]     broken=1 -> (broken'=2) & (rc'=$\gamma$);
 [up!]     broken=2 @ rc -> (inform'=2) &
                            (broken'=0) & (fc'=$\mu$);

 [f!!] inform=1 -> (inform'=0);
 [u!!] inform=2 -> (inform'=0);
endmodule
\end{lstlisting}
\caption{\small Basic Element \iosa\ symbolic model.}
\label{fig:basicelement}
\vspace{-2em}
\end{wrapfigure}
We present a symbolic language to describe an \iosa\ model.  This language is the
input language of the tool
FIG~\cite{budde2016compositional,thesis/unc/Budde2017} and has some strong
resemblance with the PRISM modelling
language~\cite{DBLP:conf/cpe/KwiatkowskaNP02}. \iosa{s} compositional style of
modelling  is also reflected in the language, where each component is modeled
separately by what we call a \emph{module}. A module is composed of a set of
variables, whose valuation represent the actual state of the component, a set of
clocks corresponding to the enabling clocks for non urgent transitions, and a
set of transitions which symbolically describe the possible jumps between states
(changes of valuations and resetting of clocks). Fig. \ref{fig:basicelement}
models a basic element as an example. Variables can be of integer (with finite
range) or boolean type.  As we will see later, also arrays can be defined as
variables. An initial value for each variable is determined after the keyword
\code{init}.  Clocks measures are defined at the transitions where they are
reset.  A transition is described by the name of the action which takes place, a
guard that defines the origin states, an enabling clock (only for the case of
non-urgent output transitions), a condition describing the target states, and
the set of clocks to be reset.
A quick overview of Fig. \ref{fig:basicelement} will help to further
understand our symbolic language: Two clocks, \code{fc} and \code{rc},
are defined at line $2$. These clocks will be used as enabling clocks
for transitions at lines $6$ and $8$, and reset on transitions at
lines $7$ and $9$ where $\gamma$ and $\mu$ are the distribution
associated with \code{rc} and \code{fc}, respectively.  Lines $3$ and
$4$ define variables \code{inform} and \code{broken}, both of integer
type ranging between $0$ and $2$, and initialized with value $0$. Line
$6$ defines a set of non-urgent output transitions, which produce the
output action \code{fl}. More precisely, this line defines the set of
non-urgent transitions
$s\trans[\{\code{fc}\},\code{fl!},\emptyset]s'$, where $s$ meets the
condition \code{broken=0}, and $s'$ is the result of changing the
values of variables \code{inform} and \code{broken} to \code{1} while
other variables remain with the same values as those in state $s$.
The \code{@} symbol precedes the enabling clock for the transition
while the \code{->} symbol distinguishes between conditions for the
origin state and the target state. The conditions on the target state
are expressed as assignments to the next values of the variables,
indicated with an apostrophe. Line $7$ defines an urgent input
transition with label \code{r}. The double question marks after the
name indicates that it describes urgent input transitions. Urgent
output transition are indicated with double exclamation marks
(\code{!!}), non urgent input transitions with a single question mark,
and non-urgent output transitions with a single exclamation mark.
At the end of line $7$ we find the reset of the clock \code{rc} to a
value from a probability distribution \code{$\gamma$}. This line then
defines transitions $s\trans[\emptyset,\code{r??},\{\code{rc}\}]s'$,
where $s$ meets with condition \code{broken=1} and $s'$ is identical to
$s$ except for variable \code{broken} which has value \code{2}.
%%, and clock \code{rc} has distribution $\gamma$.
At line $11$, an urgent output transition is defined, indicating the
failure of this component through action \code{f!!}.  We will usually
use these urgent transitions to synchronize and communicate with other
modules.

The text of Fig.~\ref{fig:basicelement} is tacitly completed with
self-loops with all inputs in all constraints that are not explicitly
written.  For example, in Fig.~\ref{fig:basicelement}, the line
``\code{[r??] broken != 1 -> ;}'' is assumed to exist.

%===============================================================================
\section{A formal syntax for RFT and its semantics} \label{sec:rftsemantics}
%===============================================================================

In this section we present a formal definition of the \rft\ similar to those of
\cite{DBLP:conf/atva/BoudaliCS07,DBLP:journals/tdsc/BoudaliCS10} along with its
semantics given in terms of \iosa. Each element of a \rft\ is characterized by
a tuple consisting of its type, its arity (i.e. number of inputs), and possibly
other parameters like probability distributions for fail and repair events in a
BE.

\begin{definition}
  Let $n,m,k\in\N^+$, and let $\mu$, $\nu$ and $\gamma$ be continuous
  probability distributions.  We define the set $\E$ of elements of a
  \rft\ to be composed of the following tuples:
  \begin{itemize}
  \item%
    $(\BE, 0, \mu, \gamma)$ and $(\SP, 0, \mu, \nu, \gamma)$, which
    represents basic and spare basic elements, with no inputs, with an
    active failure distribution $\mu$, a dormant failure distribution
    $\nu$, and a repair distribution $\gamma$.
  \item%
    $(\AND, n)$, $(\OR, n)$ and $(\PAND, n)$, which represent AND, OR
    and PAND gates with $n$ inputs, respectively,
  \item%
    $(\VOT, n, k)$, which represent a $k$ from $n$ voting gate,
  \item%
    $(\FDEP, n)$, which represents a functional dependency gate, with
    $1$ trigger input and $n-1$ dependent ones. By convention the
    first input is the triggering one.
  \item%
    $(\SG, n)$, which represents a SPARE gate with one main input and
    $n-1$ spare inputs. By convention the first input is the main one.
  \item%
    $(\RBOX, n)$, which represents a RBOX element for $n$ BEs (or SBEs).
  \end{itemize}
\end{definition}

A \rft\ is a directed acyclic graph, for which every vertex $v$ is labeled with
an element $\lab(v)\in\E$.  An edge from $v$ to $w$ means that the output of $v$
is connected to an input of $w$.  Since the order of the inputs is relevant, we
give them in terms of a list $\inputs(w)$ instead of a set.  Similarly,
$\spareinputs(v)$ will list all the spare gates to which a spare basic element
$v$ is connected as an input. Let $\type(v)$ indicate the \emph{type} of $v$.
That is, $\type(v)$ is the first projection of $\lab(v)$.  Let $\arity(v)$
indicate the number of inputs of $v$, that is, it is the second projection
$\lab(v)$.

\begin{definition}
  \label{def:rft}

  A repair fault tree is a four-tuple
  $T=(\ver,\inputs,\spareinputs,\lab)$, where $\ver$ is a set of
  vertices, $\lab\colon \ver\rightarrow\E$ is a function labeling each
  vertex with a RFT element, $\inputs\colon \ver\rightarrow\ver^*$ is
  a function assigning $\arity(v)$ inputs to each element $v$ in
  $\ver$, and $\spareinputs\colon\ver\rightarrow\ver^*$ which indicate
  which spare gates manage each spare BE.
  The set of edges $\edges = \{(v,w) \in \ver^2 \mid \exists j \cdot v
  = (\inputs(w))[j]\}$ is the set of pairs $(v,w)$ such that $v$ is an
  input of $w$.  If such an edge exists, we will say that $v$ is
  connected to $w$ and $w$ to $v$.  In addition, a RFT $T$ should
  satisfy the following conditions:
  \begin{itemize}
  \item%
    The tuple $(\ver,\edges)$ is a directed acyclic graph (DAG).
  \item%
    $T$ has a unique top element, i.e. a unique element whose non
    dummy output is not connected to another gate.  That is, there is
    a unique vertex $v\in \ver$ such that for all $w\in \ver$,
    $(v,w)\notin E$ and $\type(v)\neq \FDEP, \RBOX$.
  \item%
    An output can not be more than once the input of a same gate.
    That is, for all $1\leq j,k \leq |\inputs(w)|$ with $\inputs(w)[j]
    = \inputs(w)[k]$, we have $j=k$.
  \item%
    Since FDEP and RBOX outputs are dummy, if $(v,w)\in E$ then
    $\type(v) \notin \{\FDEP,\RBOX\}$.
  \item%
    The inputs of a repair box can only be basic elements. I.e., if
    $(v,w) \in E$ and $\type(w)=\RBOX$ then either $\type(v)=\BE$ or
    $\type(v) = \SP$.
  \item%
    Each (spare) basic element can be connected to a single RBOX. I.e.,
    if $(v,w) \in E$ and $(v,w') \in E$ and
    $\type(w)=\type(w')=\RBOX$, then $w=w'$.
  \item%
    The spare inputs of a spare gate can only be spare basic elements,
    while its main input can only be a basic element.  I.e., if
    $(v,w)\in E$ and $\type(w)=\SG$ then $\type(\inputs(v)[0])=\BE$
    and for $j > 0, \type(\inputs(v)[j])=\SP$. Furthermore, a spare
    basic element can only be connected to a spare gate or a RBOX, i.e.,
    if $(v,w)\in E$ and $\type(v)=\SP$ then
    $\type(w)\in\{\SG,\RBOX\}$.
  \item%
    A spare basic element is an input of a spare gate, if and only if
    that spare gate is spare input of the spare basic element, i.e. for
    $v$ and $v'$ such that $\lab(v')=(\SP,0,\mu,\nu,\gamma)$ and
    $\lab(v)=(\SG,n)$, $(v',v)\in E$ if and only if there exists $j$
    such that $v=\spareinputs(v')[j]$.
  \item%
    A basic element can be connected to at most one spare gate, i.e. if
    $(v,w)\in E$ and $(v,w')\in E$ with $\type(w)=\type(w')=\SG$ and
    $\type(v)=\BE$ then $w=w'$.
  \item%
    If a basic element is connected to a spare gate then it can not be
    connected to a FDEP gate, i.e. if $(v,w)\in E$ and $\type(v)=\BE$
    and $\type(w')=\SG$, then there is no $(v,w')\in E$ such that
    $\type(w')=\FDEP$.
  \end{itemize}
\end{definition}

In the following, we present a parametric semantics for \rft\ elements.  This
will be used later to obtain the semantics for each vertex in a given \rft, and
the consequent semantics of the full model as a parallel composition of its
components.  In this section, we only give the semantics for BEs, AND gates, OR
gates, PAND gates, and RBOX. Remember that FDEP can be replaced by OR gates.
Similarly, voting gates can be modeled by a series of AND and OR gates (although
a simpler model can be found in Appendix \ref{apx:kfromn}).
In the design of the \iosa\ modules we should take into account the
communication of each element of a \rft\ with its children and parents.  For
instance a basic element has to communicate its failure and repair to those
gates for which it is an input. Similarly, a RBOX has to communicate to its
inputs a \emph{start repairing} signal. In order to do so, the semantics of each
element will be given by a function, which takes actions as parameters.

\begin{wrapfigure}[]{l}{0.58\textwidth}
  \begin{lstlisting}
module AND
 informf: bool init false;
 informu: bool init false;
 count: [0..2] init 0;

 [`\cv{f$_1$}`??] count=1 -> (count'=2) & (informf'=true);
 [`\cv{f$_1$}`??] count=0 -> (count'=1);
 [`\cv{f$_1$}`??] count=2 -> ;
 [`\cv{f$_2$}`??] count=1 -> (count'=2) & (informf'=true);
 [`\cv{f$_2$}`??] count=0 -> (count'=1);
 [`\cv{f$_2$}`??] count=2 -> ;

 [`\cv{u$_1$}`??] count=2 -> (count'=1) & (informu'=true);
 [`\cv{u$_1$}`??] count=1 -> (count'=0);
 [`\cv{u$_1$}`??] count=0 -> ;
 [`\cv{u$_2$}`??] count=2 -> (count'=1) & (informu'=true);
 [`\cv{u$_2$}`??] count=1 -> (count'=0);
 [`\cv{u$_2$}`??] count=0 -> ;

 [`\cv{f}`!!] informf & count=2 -> (informf'=false);
 [`\cv{u}`!!] informu & count!=2 -> (informu'=false);
endmodule
  \end{lstlisting}
  \caption{\small AND gate \iosa\ symbolic model.}\label{fig:and}
  \vspace{-1em}
\end{wrapfigure}
For a \emph{BE} element $e\in\E$, its semantics is a function
$[\![e]\!]:\actions^5\rightarrow\iosa$, where
$[\![(\BE,0,\mu,\gamma)]\!](\tv{fl},\tv{up},\tv{f},\tv{u},\tv{r})$ results in
the \iosa\ of Fig.~\ref{fig:basicelement}.  The state of a basic element is
defined by the fail clock \code{fc}, the repair clock \code{rc}, a variable
\code{signal} that indicates when to signal the failure or repair, and variable
\code{broken} to distinguish between broken and normal states.  A basic element
fails when clock \code{fc} expires (line 6) and immediately informs it with the
urgent signal \code{\tv{f}!!} at line 11.  As soon as the repair begins by the
corresponding connected repair box (line 7), clock \code{rc} is set. When it
expires, the component becomes repaired. Hence, \code{fc} is set again at line
8, and the repair is signaled with urgent action \code{\tv{u}!!}  at line 11. At
the starting state of an \iosa\ module all its clocks are set randomly according
to their associated distributions. Thus, \code{rc} is set at the initial state
and could eventually expire without having been set by a repair transition. This
is why we have to distinguish between cases when the BE is being repaired
(\code{broken=2}) from when it is not.

For an \emph{AND gate} element with two inputs, its semantics is a function
$[\![e]\!]:\actions^{6}\rightarrow\iosa$, where
$[\![(\AND,2)]\!](\tv{f},\tv{u},\tv{f$_1$},\tv{u$_1$},\tv{f$_2$},\tv{u$_2$})$ results in
the \iosa\ in Fig.~\ref{fig:and}. At lines 6 to 11, the AND gate gets informed
of the failure of either of its inputs. Upon failure of some input, we
distinguish between the case where the other input has already failed
(\code{count=1}) and the case where it has not (\code{count=0}).  In the first
case the AND gate has to move to a failure state, for which we set the
\code{informf} variable in order to enable the signaling of failure at line 20.
Furthermore in both cases we increase the value of \code{count} so that we take
note of the failure of an input.  A similar reasoning is done for the case of
the repairing of an input at lines 13 to 18.  In this case we have to set the
module to signal a repair when an input gets repaired at a state where both
inputs were failing (lines 13 and 16), by enabling transition at line 21.
From now on, we omit writing down self loops originated by \iosa's input
enabledness, such as lines $8, 11, 15$ and $18$ as they are assumed to be there.
Nevertheless, we remark that it is necessary to take them into account when
analyzing confluence in the next section.
The semantics for an OR gate is similar to the AND gate and can be found in
Appendix \ref{app:orGate}.

The semantics of a \emph{$n$ inputs repair box with priority policy}, is a
function $[\![(\RBOX,n)]\!]:\actions^{3*n}\rightarrow\iosa$, where
$[\![(\RBOX,n)]\!](\tv{fl}_0,\tv{up}_0,\tv{r}_0,...,\tv{fl}_{n-1},\tv{up}_{n-1},$
$\tv{r}_{n-1})$ results in the \iosa\ of Fig.~\ref{fig:rbox}.  The RBOX with
priority uses the array \code{broken[n]}
\begin{wrapfigure}{l}{0.65\textwidth}
  \vspace{-2em}
  \begin{lstlisting}
module RBOX
 broken[n]: bool init false;
 busy: bool init false;

 [`\cv{fl$_0$}`?] -> (broken[0]'=true);
 ...
 [`\cv{fl$_{n-1}$}`?] -> (broken[n-1]'=true);

 [`\cv{r$_0$}`!!] !busy & broken[0] -> (busy'=true);
 ...
 [`\cv{r$_{n-1}$}`!!] !busy & broken[n-1] & !broken[n-2]
               & ... & !broken[0] -> (busy'=true);

 [`\cv{up$_0$}`?] -> (broken[0]'=false) & (busy'=false);
 ...
 [`\cv{up$_{n-1}$}`?] -> (broken[n-1]'=false) & (busy'=false);
endmodule
  \end{lstlisting}
  \caption{\small RBOX with priority policy}\label{fig:rbox}
  \vspace{-2em}
\end{wrapfigure}
to keep track of failed inputs, updating it when it receives their fail signals
(lines 5 to 7) and up signals (lines 13 to 15).  At the same time, when not
busy, it sends repair signals to broken inputs (lines 9 to 12).  Guards ensure
the priority order for repairing. Note that instead of listening to the urgent
output signals of the input BEs, it listens for the non-urgent actions of the
transitions that trigger the failure or repair. This is done with the only
purpose of facilitating the confluence analysis over this module. Other types of
repair boxes can be modeled, taking into account different repairing policies.
(see App.~\ref{apx:repbox}).

The semantics of a \emph{Priority AND gate} with $2$ inputs is defined by
$[\![(\PAND,2)]\!]:\actions^{6}\rightarrow\iosa$, where
$[\![(\PAND,2)]\!](\tv{f},\tv{u},\tv{f}_0,\tv{u}_0,\tv{f}_1,\tv{u}_1)$ results
in the \iosa\ of Fig.~\ref{fig:pand}. PAND gates fail only when their inputs
fail from left to right. This allows to condition the failure of a system not
only to the failure of the subsystems but also to the ordering in which they
fail. Notice that an $n$ inputs PAND gate is simply a syntax sugar for a system
of $n-1$ two-input PAND gates connected in cascade. Literature is not always
clear or even disagrees on what should be the behavior of the PAND gate in case
both inputs fail at the same time
\cite{manian1999bridging,DBLP:conf/issre/CoppitSD00}. This situation arises in
some constructions with AND and OR gates, or when the inputs of a PAND gate are
connected to the a same FDEP (see Fig. \ref{fig:nondet}). Some proposals
disallow these situations and discard them on early syntactic
checks~\cite{DBLP:conf/safecomp/RuijtersRBS17}.
\begin{wrapfigure}[]{l}{0.59\textwidth}
  \vspace{-2em}
  \begin{lstlisting}
module PAND

 f1: bool init false;
 f2: bool init false;
 st: [0..4] init 0; \\ up, inform fail, failed,
                       inform up, unbreakable

 [_?] st=0 & f1 & !f0 -> (st'=4); `\label{line:joker}`

 [`\cv{f$_0$}`??] st=0 & !f0 & !f1-> (f0'=true);
 [`\cv{f$_0$}`??] st=0 & !f0 & f1 -> (st'=1) & (f0'=true);
 [`\cv{f$_0$}`??] st!=0 & !f0     -> (f0'=true);
 [`\cv{f$_0$}`??] f0              -> ;

 [`\cv{f$_1$}`??] st=0 & !f0 & !f1         -> (f1'=true);
 [`\cv{f$_1$}`??] st=0 & f0 & !f1          -> (st'=1) & (f1'=true);
 [`\cv{f$_1$}`??] st=3 & !f1               -> (st'=2) & (f1'=true);
 [`\cv{f$_1$}`??] (st==1|st==2|st=4) & !f1 -> (f1'=true);
 [`\cv{f$_1$}`??] f1                       -> ;

 [`\cv{u$_0$}`??] st!=1 & f0 -> (f0'=false);
 [`\cv{u$_0$}`??] st=1  & f0  -> (st'=0) & (f0'=false);
 [`\cv{u$_0$}`??] !f0 -> ;

 [`\cv{u$_1$}`??] (st=0|st=3) & f1 -> (f1'=false);
 [`\cv{u$_1$}`??] (st=1|st=4) & f1 -> (st'=0) & (f1'=false);
 [`\cv{u$_1$}`??] st=2 & f1 -> (st'=3) & (f1'=false);

 [`\cv{f}`!!] st=1 -> (st'=2);
 [`\cv{u}`!!] st=3 -> (st'=0);

endmodule
  \end{lstlisting}
  \caption{\small PAND gate.}\label{fig:pand}
  \vspace{-2em}
\end{wrapfigure}
Some others assume a non-deterministic situation and find it important to
analyze scenarios where the behavior is in fact unknown
\cite{DBLP:conf/atva/BoudaliCS07}. Other works decided that the PAND gate does
not fail unless its inputs break strictly from left to right
\cite{DBLP:conf/dsn/BoudaliCS07,bobbio2004parametric}. Some others state that
PAND gates also fail when both their inputs fail at the same time
\cite{DBLP:conf/issre/CoppitSD00,DBLP:journals/tr/BoudaliD08,DBLP:journals/ress/BoudaliD05}.
We opted for this last case, so the gates needs to be able to identify if time
has passed between the occurrence of the failures, and act consequently. In the
particular case where no time passes between the failure of the inputs, we
consider that the order in which the dependent BEs fail does not really matter
and thus the non-determinism is spurious. To identify if time has passed between
the occurrence of the input failures, the model listens to any output actions,
which indicate that a clock has expired.

\begin{wrapfigure}[11]{l}{0.45\textwidth}
  \vspace{-2em}
  \center
  \noindent\includegraphics[scale=.35]{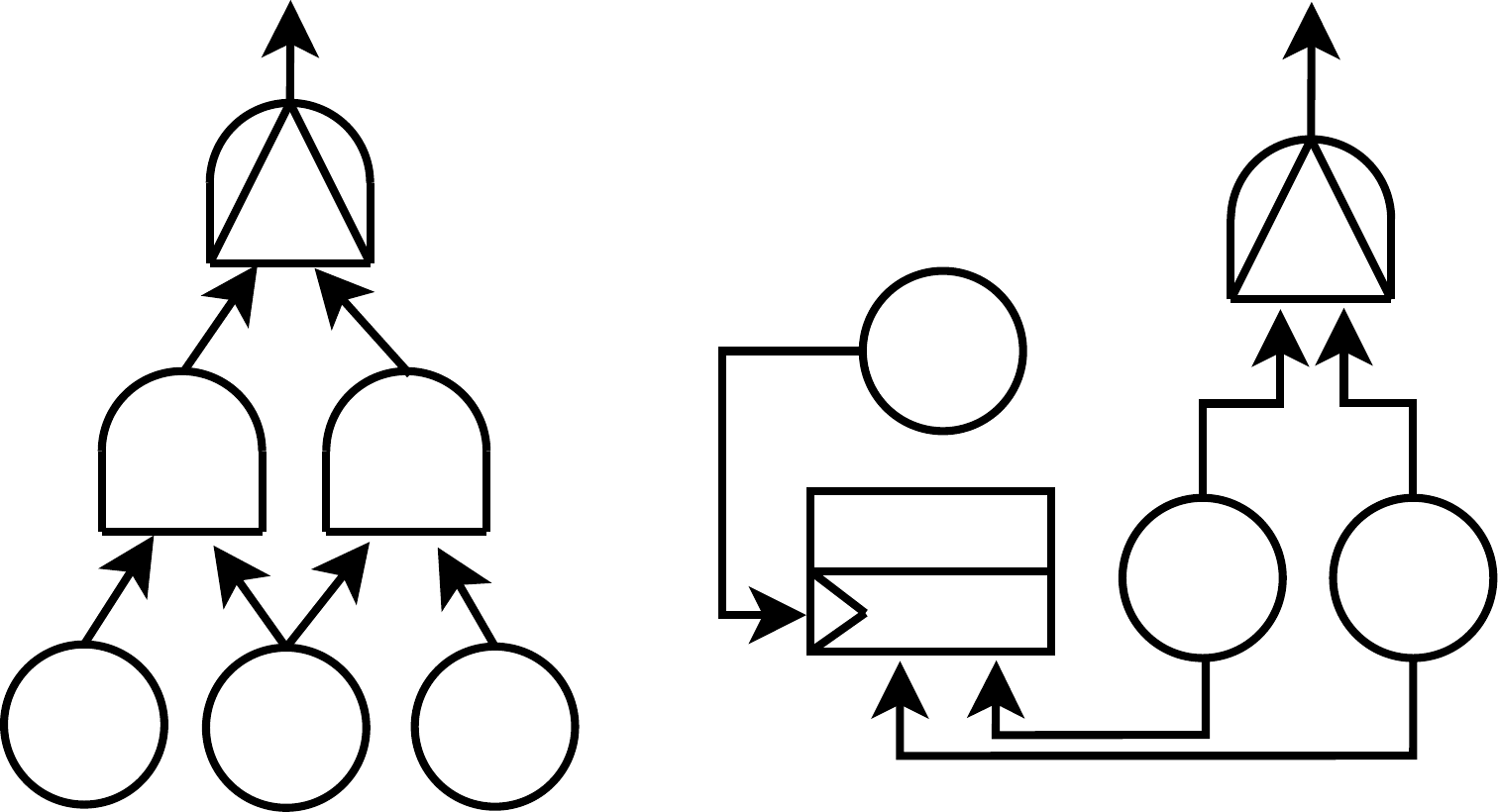}
  \caption{\small Spurious non-determinism.}
  \label{fig:nondet}
\end{wrapfigure}
This is done by a special input action at line \ref{line:joker}, which
synchronizes with all non-urgent outputs, regardless the name of the action.
Notice that there is only one scenario that we want to rule out, which is when
the second input fails and then time passes without the first input failing too.
This is in fact the case described by the guard of line \ref{line:joker}.
Furthermore, this transition moves to the `unbreakable' state, from which it can
only go back when input 1 is fixed.  In consequence, the failure of the gate
occurs either if both inputs fail at the same time or if the first input fails,
then time passes, and then the second input fails. \remarkREM{we can remove the
last sentence if space is required}

The semantics of a \rft\ is that of the parallel composition of the semantics of
its components, being conveniently synchronized.

\begin{definition}  \label{def:rftsemantic}
  Given a \rft\ $T=(\ver,\inputs,\spareinputs,\lab)$ we define the semantics of
  $T$ as $[\![T]\!] = ||_{v\in V}[\![v]\!]$ where $[\![v]\!]$ is defined by:

  \[
    [\![v]\!]=
    \begin{cases}
      [\![\lab(v)]\!](\text{\rm\tv{fl}}_v,\text{\rm\tv{up}}_v,\text{\rm\tv{f}}_v,\text{\rm\tv{u}}_v,\text{\rm\tv{r}}_v)
      & \text{if } \lab(v)=(\BE,0,\mu,\gamma)
      \\
      [\![\lab(v)]\!](\text{\rm\tv{f}}_v,\text{\rm\tv{u}}_v,\text{\rm\tv{f}}_{i(v)[0]},\text{\rm\tv{u}}_{i(v)[0]},...,\text{\rm\tv{f}}_{i(v)[n-1]},\text{\rm\tv{u}}_{i(v)[n-1]})\hspace{-17em}
      \\ &
      \text{if } \lab(v)\in\{(\AND,n),(\OR,n)\}
      \\
      [\![\lab(v)]\!](\text{\rm\tv{f}}_v,\text{\rm\tv{u}}_v,\text{\rm\tv{f}}_{i(v)[0]},\text{\rm\tv{u}}_{i(v)[0]},\text{\rm\tv{f}}_{i(v)[1]},\text{\rm\tv{u}}_{i(v)[1]})
      &
      \text{if } \lab(v)=(\PAND,2)
      \\
      [\![\lab(v)]\!](\text{\rm\tv{fl}}_{i(v)[0]},\text{\rm\tv{up}}_{i(v)[0]},\text{\rm\tv{r}}_{i(v)[0]},...,\text{\rm\tv{fl}}_{i(v)[n-1]}, \text{\rm\tv{up}}_{i(v)[n-1]},\text{\rm\tv{r}}_{i(v)[n-1]})\hspace{-17em}
      \\ &
      \text{if } \lab(v)=(\RBOX,n)
    \end{cases}
    \]
\end{definition}

In Section~\ref{sec:extendedSemantics}, we extend the semantics to spare gates
and spare basic elements.

%===============================================================================
\section{\rft{s} are weakly deterministic}\label{sec:weakDeterminism}
%===============================================================================

In this section we show that \rft{s} composed only by BEs, AND gates, OR gates,
PAND gates, and RBOX, are weakly deterministic. Since voting and FDEP gates can
be constructed using OR and AND gates, the result extends to these gates.
Results in this section rely heavily on results about weak determinism on \iosa\
proven in~\cite{DArgenioMonti18}. Therefore, we first summarize the essentials
of~\cite{DArgenioMonti18} for this paper.

\begin{definition}
\label{def:confluence}
  An \iosa\ is \emph{confluent} if for all pair of urgent actions $a$
  and $b$, and for every (reachable) state $s$, it satisfies that, if
  $s\trans[\emptyset,a,C_1]s_1$ and $s\trans[\emptyset,b,C_2]s_2$,
  then there is a state $s_3$ such that $s_2\trans[\emptyset,a,C_1]s_3$
  and $s_1\trans[\emptyset,b,C_2]s_3$.
\end{definition}

Note that, according to this definition, regardless the order of the confluent
transitions, the same state is reached. This non-determinism is spurious in the
sense that it does not alter the stochastic properties of the given \iosa,
regardless the manner it is solved. Since non-determinism can only arise on
urgent actions, we say that a \emph{closed} \iosa\ is \emph{weakly
deterministic} if all its urgent actions are confluent.
In~\cite{DArgenioMonti18}, we provided sufficient conditions to ensure that a
closed \iosa\ is weakly deterministic. This is stated in
Theorem~\ref{theo:sufCondForDet} below which requires the following definition.

\begin{definition}
  Given an \iosa\ $\I$ with state space $\states$ and actions
  $\actions$, we distinguish the following sets of actions:
  \begin{itemize}
  \item%[\emph{Initially enabled actions:}]%
    A set of urgent output actions $B \subseteq \outactions \cap\coactions$ is
    \emph{initial} if each $b\in B$ is enabled in $s_0$, i.e. if for
    each $b\in B$ there is a state $s\in \states$ and
    $C\subseteq\clocks$, such that $s_0\trans[\emptyset,b,C]s$.
  \item%[\emph{Spontaneously enabled actions:}]%
    We say that a set $B\subseteq\outactions\cap\coactions$ of output
    urgent actions is \emph{spontaneously enabled} by
    $b\in\actions\setminus\coactions$ if there are stochastically
    reachable states $s,s'\in\states$ (a state is stochastically
    reachable if there is a path in the \iosa\ from the initial state
    that reaches such state with probability greater than zero) such
    that $s$ is stable, $s\trans[\_,b,\_]s'$, and all actions in $B$
    are enabled in $s'$.
  \item%[\emph{Triggered actions and the triggering relation:}]
    \label{def:trigrel}%
    Let $\act\in \coactions$ and $\actb\in \outactions\cap\coactions$.
    We say that $\act$ \emph{triggers} $\actb$ if there are
    stochastically reachable states $s_1,s_2,s_3\in\states$ such that:
    $s_1\trans[\_,a,\_]s_2$, $s_2\trans[\_,b,\_]s_3$, and, if
    $\act\neq\actb$, then there is no outgoing transition from $s_1$ labeled
    with $b$. %s_1\ntrans[\_,b,\_]$.
    The set $\{(a,b)\mid a \text{ triggers } b\}$ is called the
    \emph{triggering relation}.
  \end{itemize}
\end{definition}

The \emph{approximate indirect triggering relation} of a composite \iosa\ is
defined as the reflexive transitive closure of the union of the triggering
relations of its components.
The following theorem from~\cite{DArgenioMonti18}, gives necessary conditions
for a closed \iosa\ not to be confluent. As a consequence, it provides
sufficient conditions for a closed \iosa\ to be \emph{weakly deterministic}.

\begin{theorem}
\label{theo:sufCondForDet}
  Given a closed composite \iosa\ $\I = \left( \I_1 \pll \ldots \pll
  \I_n \right)$ with actions $\actions$, if $\I$ is not confluent then
  there exist a pair of urgent actions $\act,\actb \in \coactions$ such that
  \begin{enumerate}
  \item%
    one of the components is not confluent with respect to $\act$ and
    $\actb$,
  \item%
    there are actions $c$ and $d$ that approximately indirectly
    trigger both $a$ and $b$, respectively, and
  \item%
    one of the following hold:
    %% \begin{enumerate}
    \begin{inparaenum}[(i)]
    \item%
      $c$ and $d$ are initial actions, or
    \item%
      there exists an action $e$ and possible empty sets $B_1$ to
      $B_n$ spontaneously enabled by $e$ in $\I_1$ to $\I_n$
      respectively, such that $c$ and $d$ are in $\bigcup_{i=1}^n
      B_i$.
    \end{inparaenum}
    %% \end{enumerate}
  \end{enumerate}
\end{theorem}

In the following, we prove accessory propositions to eventually prove, using
Theorem~\ref{theo:sufCondForDet}, that the \iosa\ defined by a \rft\ is weakly
deterministic.

\begin{proposition}
  \label{prop:initandspont}
  Let $T$ be a \rft.  $[\![T]\!]$ has no initially enabled actions.
  Moreover, the only spontaneous sets of actions are singletons of the
  form $\{\text{\rm\code{f}}_v\}$ and $\{\text{\rm\code{u}}_v$\}, for
  $\type(\lab(v))=\BE$, which are spontaneously enabled by
  $\text{\rm\code{fl}}_v$ and $\text{\rm\code{up}}_v$, respectively.
\end{proposition}
\begin{proof}
  As a consequence of \cite{DArgenioMonti18}, the initially enabled
  actions of $[\![T]\!]$ are contained in the union of the sets of
  initially enabled actions of its components $[\![v]\!]$, $v\in V$,
  and the spontaneously enabled actions of $[\![T]\!]$ are contained
  in the union of the spontaneously enabled sets of $[\![v]\!]$.  It
  is direct to see that, for any element $e\in\E$, none of the urgent
  outputs are enabled at the initial state of $[\![e]\!]$, since their
  guards are initially false.  Furthermore, the only non-urgent output
  transition in our models are at lines 6 and 8 of the BE
  (Fig.~\ref{fig:basicelement}). Let $v\in V$ such that
  $\type(v)=\BE$. Then, after taking transition at line 6 the only
  urgent output enabled is $\code{f}_v$ (on the instance $[\![v]\!]$),
  while after taking transition at line 8 the only one is
  $\code{u}_v$, and thus these are the only possible spontaneous
  enabled actions. \qed
\end{proof}

\begin{proposition}
\label{prop:nonconfluent}
  Let $T$ be a \rft. The only possible pairs of non-confluent actions
  in $[\![T]\!]$ are
  $\{(\text{\rm\tv{f}}_v,\text{\rm\tv{u}}_{v'})\mid v,v' \in i(w), t(w)\in\{\AND, \OR, \PAND\}\} \cup
  \{(\text{\rm\tv{f}}_{w},\text{\rm\tv{u}}_{v}),(\text{\rm\tv{u}}_{w},\text{\rm\tv{f}}_{v})\mid v \in i(w), t(w)\in\{\AND, \OR\}\}$.
\end{proposition}
\begin{proof}
  The proof of this Proposition follows an exhaustive check over each urgent
  transition of each model, in order to single out any non-confluent situation,
  and can be found at Appendix \ref{appdx:proofOfNonConfluent}
\end{proof}

\begin{proposition}
  \label{prop:triggering}
  Let $T$ be a \rft. For each $v\in V$, the triggering relation of
  $[\![v]\!]$ is given by:
  \begin{itemize}
    \item $\{\}$, if $\lab(v)\in\{(\BE,0,\mu,\gamma),(\RBOX,n)\}$,
    \item $\{(\text{\rm\tv{f}}_w,\text{\rm\tv{f}}_v)\mid w\in i(v)\} \cup
    \{(\text{\rm\tv{u}}_w,\text{\rm\tv{u}}_v)\mid w\in i(v)\}$, if $\lab(v)\in\{(\AND,n),(\OR,n)\}$, and
    \item $\{(\text{\rm\tv{u}}_w,\text{\rm\tv{u}}_v)\mid w = i(v)[1]\} \cup
          \{(\text{\rm\tv{f}}_w,\text{\rm\tv{f}}_v)\mid w\in i(v)\}$, if $\lab(v)=(\PAND,2)$.
  \end{itemize}
\end{proposition}
\begin{proof}[sketch]
  It sufficies to make a satisfiability analysis over guards and
  postconditions of each pair $(t_a,t_b)$ with $t_b$ an output urgent
  symbolic transition and $t_a$ any urgent symbolic transition, taking
  into account only reachable states. \qed
\end{proof}

\begin{theorem}
  Let $T$ be a \rft.  Then $[\![T]\!]$ is weakly deterministic.
\end{theorem}
\begin{proof}
  We look for $a,b,c,d$ and $e$ as well as sets $B_i$ with $i=1\dots
  n$ as Theorem~\ref{theo:sufCondForDet} suggests. Since
  Prop.~\ref{prop:initandspont} ensures that there are no initially
  enabled actions in $[\![T]\!]$, $c$ and $d$ should be spontaneously
  enabled actions. By the same proposition either $e$ is of the form
  $\code{fl}_v$ for some $v$ and then
  $\bigcup_{i=1}^1B_i=B_1=\{\code{f}_v\}$,
  or $e$ is of the form $\code{up}_v$ for some $v$ and then
  $\bigcup_{i=1}^1B_i=B_1=\{\code{u}_v\}$.
  In the first case, we get $c=d=\code{f}_v$ for some $v$, and in the
  second case $c=d=\code{u}_v$. Furthermore, by
  Prop.~\ref{prop:nonconfluent}, either $a$ is of the form
  $\tv{f}_w$ for some $w$ and $b$ is of the form $\tv{u}_{w'}$ for
  some $w'$ or the other way around. As shown by
  Prop.~\ref{prop:triggering}, fail actions ($\tv{f}_v$ for some
  $v$) only trigger fail actions and up actions ($\tv{u}_v$ for some
  $v$) only trigger up actions, thus it is impossible that $c$ and $d$
  indirectly trigger $a$ and $b$ respectively.
  Therefore, it is not possible to find actions $a$, $b$, $c$, $d$,
  and $e$ satisfying conditions 1 to 3 in
  Theorem~\ref{theo:sufCondForDet}, and hence $[\![T]\!]$ is
  confluent. Since $[\![T]\!]$ is also closed, then it is weakly
  deterministic.
  \qed
\end{proof}

%===============================================================================
\section{An extended Semantics}\label{sec:extendedSemantics}
%===============================================================================

In this section we add the spare gate and spare basic element to the semantics
of \rft{s}. As before, we aim to guarantee that the \iosa\ model derived from
the \rft\ is weakly deterministic. In order to do so, we need to bring special
attention to two particular scenarios that could introduce non-determinism if
not correctly tackled.

The first scenario is given when a main basic element fails at a spare gate
which is served with several spare basic elements. At this point, it arises the
question of which of the available spare basic elements should the spare gate
take. Traditionally, spare elements are selected in order from an ordered set.
To generalize this mechanism for the selection of the spares we intend to allow
for more complex state-involved policies. It should be always the case that
this policy is deterministic in its elections.
The second scenario arises when several spare gates have requested a broken or
already taken SBE, which eventually gets fixed by a repair box or freed by the
owning spare gate. At this point, it is unclear which of the requesting spare
gates will take the newly available SBE. For this, we define sharing policies on
the SBE. Thus, to provide semantics to an SBE, we actually introduce two \iosa\
modules: one providing the extended behavior of a BE that can be taken from
dormant to enabled state and vice versa, and another one, the \emph{multiplexer}
module, which manages the sharing of the SBE. Notice that this scenario is not a
problem in the absence of repair boxes, since in such cases SBEs do not become
available after they are taken or fail. It is neither a problem when spare
elements are not shared by different spare
gates~\cite{bobbio2004parametric,DBLP:journals/tse/BobbioFGP03}. The
work~\cite{DBLP:conf/dsn/JungesGKS16} also studies race conditions in spare
gates when two spare gates fail at the same time. This last situation is
impossible in our settings given the last two properties of Definition
\ref{def:rft} and the fact that two simultaneous failures of our basic elements
is discarded by the IOSA deterministic semantics.

The models for the spare gate, the spare basic element and the multiplexer can
be found in Appendix \ref{app:spare}. We extend the semantics of the \rft\ with
the SBE and SG elements as follows.

\begin{definition}
  \label{def:rftextsemantic}
  Given a \rft\ $T=(V,E)$, we extend Definition \ref{def:rftsemantic}
  with the following cases:
  \[
    [\![v]\!]=
    \begin{cases}
      \cdots\\
      [\![ \lab(v) ]\!](\text{\rm\tv{fl}}_v,\text{\rm\tv{up}}_v,\text{\rm\tv{f}}_v,\text{\rm\tv{u}}_v,\text{\rm\tv{r}}_v,\text{\rm\tv{e}}_v,\text{\rm\tv{d}}_v,\text{\rm\tv{rq}}_{(\spareinputs(v)[0],v)},\text{\rm\tv{asg}}_{(v,\spareinputs(v)[0])},\hspace{-10em}\\
      \phantom{[\![ \lab(v) ]\!](}
      \text{\rm\tv{rel}}_{(\spareinputs(v)[0],v)},\text{\rm\tv{acc}}_{(\spareinputs(v)[0],v)},\text{\rm\tv{rj}}_{(v,\spareinputs(v)[0])},..,\text{\rm\tv{rj}}_{(v,\spareinputs(v)[n-1])})\hspace{-10em}\\
      &\text{if } \lab(v)=(\SP,n,\mu,\nu,\gamma)\\
      [\![\lab(v)]\!](\text{\rm\tv{f}}_v,\text{\rm\tv{u}}_v,\text{\rm\tv{fl}}_{i(v)[0]},\text{\rm\tv{up}}_{i(v)[0]},\text{\rm\tv{fl}}_{i(v)[1]},\text{\rm\tv{up}}_{i(v)[1]},\text{\rm\tv{rq}}_{(v,i(v)[1])},\text{\rm\tv{asg}}_{(i(v)[1],v)},\hspace{-10em}\\
      \phantom{[\![ \lab(v) ]\!](}
      \text{\rm\tv{acc}}_{(v,i(v)[1])},\text{\rm\tv{rj}}_{(i(v)[1],v)},\text{\rm\tv{rel}}_{(v,i(v)[1])},...,\text{\rm\tv{rel}}_{(v,i(v)[n-1])})\hspace{-10em}\\
      &\text{if } \lab(v)=(\SG,n)
    \end{cases}
  \]
\end{definition}

Notice that in the case of the SBE and SG, several signals are indexed by a pair
of elements. This pair indicates which gate performs the action and which one
listens for synchronisation.  As an example, $\tv{asg}_{(v,\spareinputs(v)[0])}$
indicates that the multiplexer that manages $v$, assigns its spare basic element
to its first connected spare gate ($\spareinputs(v)[0]$).

Unfortunately, we could not find an easy or direct way to prove that this
extension is indeed weakly deterministic, as we did with the RFT without spares.
This is due in part to the complexity of the \iosa\ modules, intended to avoid
the aforementioned non-deterministic situations.
While the spare basic element module can be easily proved to be confluent, this
is not the case for the modules of the multiplexer and the spare gate. When
analyzing these modules in isolation we find that some transitions are not
confluent and Theorem~\ref{theo:sufCondForDet} could not be used directly.
However, by partially composing spare gates with multiplexers, we were able to
check that conditions of Theorem~\ref{theo:sufCondForDet} are not met. We
automatically perform this check in several configurations, and showed that they
are confluent.  As parallel composition preserves confluence, they can be
inserted in other \rft\ contexts yielding weakly deterministic \iosa{s}.
\footnote{For the reviewers eyes, only: the scripts that prove said
  configurations are available at
  \url{https://git.cs.famaf.unc.edu.ar/raulmonti/DeterminismScriptsRFT}.}

%===============================================================================
\section{Conclusion}\label{sec:conclusion}
%===============================================================================

%
In this work we have defined a semantics for Dynamic Fault Trees with repair box in
terms of Input/Output Stochastic Automata, introducing the novel feature of
general probability measures for failure and repair rates of basic elements.
Furthermore we have shown that our semantics produces weakly deterministic models
which are hence amenable for discrete event simulation. In particular, our
models serve as direct input to the FIG Simulator
(\url{http://dsg.famaf.unc.edu.ar/fig})~\cite{budde2016compositional,thesis/unc/Budde2017},
as well as other tools through the intermediate language
Jani~\cite{DBLP:conf/tacas/BuddeDHHJT17}.
% A future work direction could be
% introducing Fail Dependency Gates (like in
% \cite{DBLP:conf/dsn/BoudaliCS07,DBLP:conf/safecomp/RuijtersRBS17}) in a way that
% they do not produce non-determinism (at the moment no order in the dependent
% failures has been defined and therefore the non-determinism is intrinsic to the
% definition).
%
A future work direction could be introducing maintenance mechanism and levels of
degradation as in \cite{ruijters2016fault}, in order to increase the
possibilities for defining repair models. Another line of work would be defining
an automatic translation from a graphical modelling tool for fault trees into the
\iosa\ models, in order to automate and ease the modelling and analysis of
industrial size systems. Adding support for spare sub-trees such as in \cite{}
would be an interesting upgrade too, also along with support for sub-tree 
dedicated repair boxes. \remarkREM{This last two things were a complaint from one
of the reviewers at QEST. I am not sure that this would mean an improvement
to this work nor an interesting thing to have. It looks like a simple 
technicality and maybe (s)he used it just to make the review look interesting.}

%%%%%%%%%%%%%%%%%%%%%%%%%%%%%%%%%%%%%%%%%%%%%%%%%%%%%%%%%%%%%%%%%%%%%%%%%%%%%%%

%%%%%%%%%%%%%%%%%%%%%%%%%%%%%%%%%%%%%%%%%%%%%%%%%%%%%%%%%%%%%%%%%%%%%%%%%%%%%%%

\newpage
\appendix

%%%%%%%%%%%%%%%%%%%%%%%%%%%%%%%%%%%%%%%%%%%%%%%%%%%%%%%%%%%%%%%%%%%%%%%%%%%%%%%

\section{\iosa\ Symbolic Language}
\label{appx:iosasymlang}
The following context free grammar defines the complete \iosa\ symbolic  modelling
language. Here * stands for as  many times as you want, + for at least one
time, ? for optional, $|$ for option, and parentheses group productions and
elements.

\begin{figure}[H]
  \begin{align*}
    \text{MODEL}      &= (\text{MODULE})+\\
    \text{MODULE}     &= (\text{VARIABLE}~|~\text{ARRAY}~|~\text{CLOCK})+
                         ~~\text{TRANSITION}+\\
    \text{VARIABLE}   &= \text{NAME} ~\code{:}~ \text{TYPE} ~\code{init}~
                         \text{VALUE} ~\code{;}\\
    \text{ARRAY}      &= \text{NAME}\code{[}\text{INT}\code{]:}~
                         \text{TYPE}~\code{init}~\text{VALUE}~\code{;}\\
    \text{CLOCK}      &= \text{NAME} ~\code{: clock ;}\\
    \text{TRANSITION} &= \code{[}~(\text{NAME}~(\code{?}|\code{??}|\code{!}|
                         \code{!!})?)?~\code{]}~\text{PRE}~(\code{@}~
                         \text{NAME})?~\code{$\trans$}~\text{POS}~\code{;}\\
    \text{PRE}        &= ((\text{NAME}=\text{EXPR}) (\code{\&}~
                         \text{NAME}=\text{EXPR})*)? \\
    \text{POS}        &= (\code{(}~\text{NAME}\code{'}=\text{EXPR}~\code{)}
                         (\code{\&}~ \code{(}~\text{NAME}\code{'}=
                         \text{EXPR}~\code{)})*)?\\
    \text{EXPR}       &= \text{VALUE}~|~\text{NAME}~|~\text{EXPR}~\text{OP}~
                         \text{EXPR}~|~\code{(}~\text{EXPR}~\code{)}~|~
                         !~\text{EXPR}\\
    \text{OP}         &= \code{|}~|~\code{\&}~|~\code{+}~|~\code{-}~|~
                         \code{*}~|~\code{/}~|~\code{=}\\
    \text{NAME}       &= (\code{a}|\code{b}|...|\code{z}|\code{A}|
                         \code{B}|...|\code{Z})(\code{a}|\code{b}|...
                         |\code{z}|\code{A}|\code{B}|...|\code{Z}|
                         \code{1}|...|\code{9}|\code{\_}|\code{-})*\\
    \text{TYPE}       &= \code{boolean}~ | ~\code{[} ~\text{INT}~ \code{..}
                         ~\text{INT}~ \code{]}\\
    \text{VALUE}      &= \code{true}~ | ~\code{false}~ | ~\text{INT}\\
    \text{INT}        &= (\code{1}|\code{2}|...|\code{9})(\code{0}|
                         \code{1}|...|\code{9})*
  \end{align*}
\caption{\iosa\ symbolic language grammar.}
\label{def:iosasymbolicgrammar}
\end{figure}

An \iosa\ model is composed by a set of modules, each one describing a
concurrent component of the system to model. The body of a module can be clearly
divided into three parts: the variables declarations, the clocks declarations,
and the transitions specification. Arrays are declared along with variables,
with the additional requirement of defining the range of the array between
brackets. Transitions preconditions are boolean formulas describing the origin
states for the symbolic transition. In this case the $\&$ symbol stands for the
logical conjunction operator while $|$ stands for the logical disjunction
operator. Postconditions on the other side, describe the changes on the module's
variables (state) by means of assignments to future values. Each assignment is
enclosed by parenthesis, and the variable's name is followed by an apostrophe to
indicate that corresponds to the value of the variable in the reached state
after taking the transition. An $\&$ separates each  assignment. Notice the
similarity with PRISM \cite{DBLP:conf/cpe/KwiatkowskaNP02} syntax for describing
transitions. Along with the assignment of values to future variables, we find
the reset of clocks. A clock is assigned a probability distribution
($clock'=\gamma$) to indicate that it will be reset to a value from that
probability distribution immediately before reaching the new state.

\section{OR Gate}
\label{app:orGate}

For an \emph{OR gate} element with two inputs, its semantics is a function
$[\![]\!]:\actions^{6}\rightarrow\iosa$, where
$[\![ (\OR,2) ]\!](f,u,f_1,u_1,f_2,u_2)$
results in the following \iosa:
  \begin{lstlisting}
module OR
  informf: bool init false;
  informu: bool init false;
  count: [0..2] init 0;

  [`\cv{f$_1$}`??] count=0 -> (count'=1) & (informf'=true);
  [`\cv{f$_1$}`??] count=1 -> (count'=2);
  [`\cv{f$_2$}`??] count=0 -> (count'=1) & (informf'=true);
  [`\cv{f$_2$}`??] count=1 -> (count'=2);

  [`\cv{u$_1$}`??] count=2 -> (count'=1);
  [`\cv{2$_1$}`??] count=1 -> (count'=0) & (informu'=true);
  [`\cv{u$_2$}`??] count=2 -> (count'=1);
  [`\cv{u$_2$}`??] count=1 -> (count'=0) & (informu'=true);

  [`\cv{f}`!!] informf & count!=0 -> (informf'=false);
  [`\cv{u}`!!] informu & count=0 -> (informu'=false);
endmodule
  \end{lstlisting}

In the OR gate model, a counter (\code{count}) is used to register how many
inputs have failed at each moment. The failing of an input increases the
counter, while the repair of an input decreases the counter. We of course take
as a premise that an input will not break two times in a row without being
repaired in the middle, neither will it be repaired if it has not failed. When
the counter changes its value from 0 to 1, the gate has to inform a failure. It
does so in transition at line 16, which gets enabled by the change of variable
\code{informf} either at line 6 or 8. In the same way, when \code{count} becomes
$0$, the repair is informed by enabling transition at line $17$ through the
change of variable \code{informu} either at line $12$ or $14$.

\section{Repair BOXes}
\label{apx:repbox}

For a \emph{repair box with first come first serve policy} element $e\in\E$ with
$n$ inputs, its semantics is a function
$[\![e]\!]:\actions^{3*n}\rightarrow\iosa$, where $[\![ (\RBOX,n)
]\!](\tv{fl}_0,\tv{up}_0,\tv{r}_0,...,\tv{fl}_{n-1},\tv{up}_{n-1},$
$\tv{r}_{n-1})$ results in the following \iosa:
  \begin{lstlisting}
module RBOX % with first come first serve policy
  queue[n]: [0..n] init 0;
  busy: bool init false;
  r: [0..n] init n;
  dummy: [0..0] init 0;

  [`\cv{fl$_0$}`?] -> (dummy'=broken(queue,0));
  ...
  [`\cv{fl$_{n-1}$}`?] -> (dummy'=broken(queue,n-1));

  [!!] fstexclude(queue,0) != -1 & r = n -> (r'=maxfrom(queue,0));

  [`\cv{r$_0$}`!!] !busy & r = 0 -> (busy'=true) & (queue[0]'=0);
  ...
  [`\cv{r$_{n-1}$}`!!] !busy & r = n-1 -> (busy'=true) & (queue[n-1]'=0);

  [`\cv{up$_0$}`?] -> (queue[0]'=0) & (busy'=false) & (r' = n);
  ...
  [`\cv{up$_{n-1}$}`?] -> (queue[n-1]'=0) & (busy'=false) & (r' = n);
endmodule
  \end{lstlisting}
The model for a \emph{repair box with first come first serve policy}  uses an
array to mark down each broken input. Notice that each position in the queue
corresponds to each input. A value $0$ on an index $i$ means that the input $i$
has not failed, while a greater value on that position indicates for ``how
long'' has it been broken. Repair boxes use some syntactic elements present in
FIG (http://dsg.famaf.unc.edu.ar/fig) simulator. These elements do not introduce
a new semantics behavior and are there only to reduce the complexity and
obfuscation that would represent modelling this using only the grammar presented
at App. \ref{appx:iosasymlang}. Examples of this are the function \code{broken}
which given an array, in this case \code{queue}, and an index, in this case
\code{0}, it increases by one the value at that index and every other value
greater than $0$ in the array. In this way we can check the order in which the
inputs failed by comparing the values at the corresponding index. The greater
the value the sooner they broke. The syntactic function \code{fstexclude} on
the other hand, takes an array and a value and returns the index of the first
element with a different value to the one passed. In this case we use it to
check if there is any failed input. If there is at least one, then
\code{maxfrom} function will return the index of the highest value in
\code{queue}, which corresponds to the input who broke first in between all the
broken ones. For a quick determinism analysis we point out that all
\code{broken}, \code{fstexclude}, and \code{maxfrom} are deterministic.
Furthermore all pairs of urgent transitions in the model are confluent given
that their preconditions are mutually exclusive given the value of variable
\tv{r}.

For a \emph{repair box with random policy} element $e\in\E$ with $n$ inputs, its
semantics is a function $[\![e]\!]:\actions^{3*n}\rightarrow\iosa$, where $[\![
(\RBOX,n) ]\!](\tv{fl}_0,\tv{up}_0,\tv{r}_0,...,\tv{fl}_{n-1},\tv{up}_{n-1},$
$\tv{r}_{n-1})$ results in the following \iosa:
  \begin{lstlisting}
module RBOX % with random policy
  broken[n]: bool init false;
  busy: bool init false;
  r: [0..n] init n;

  [`\cv{fl$_0$}`?] -> (broken[0]'=true);
  ...
  [`\cv{fl$_{n-1}$}`?] -> (broken[n-1]'=true);

  [!!] some(broken) & r = n -> (r'=random(broken));

  [`\cv{r$_0$}`!!] !busy & r = 0 -> (busy'=true);
  ...
  [`\cv{r$_{n-1}$}`!!] !busy & r = n-1 -> (busy'=true);

  [`\cv{up$_0$}`?] -> (broken[0]'=false) & (busy'=false) & (r' = n);
  ...
  [`\cv{up$_{n-1}$}`?] -> (broken[n-1]'=false) & (busy'=false) & (r' = n);
endmodule
  \end{lstlisting}
The model for a \emph{random policy repair box} presents two new syntactic
elements from FIG. These are the function \code{some}, which returns a boolean
value indicating if there is some value different to zero in the array, and the
function \code{random}, which models an uniform selection of an index between
the non zero valued positions at an array. Given that these two functions are
deterministic, and with a similar analysis as for the first come first serve
policy repair box, we can deduce that this is also a deterministic model.

\section{Voting gate}
\label{apx:kfromn}

The following \iosa\ model corresponds to the modelling of a 2 from 3 voting
gate. A generalisation to other values of $N$ and $K$ can be easily obtained.
  \begin{lstlisting}
module VOTING
  count: [0..3] init 0;
  inform: bool init false;

  [`\cv{f$_0$}`??] -> (count'=count+1) & (inform'=(count+1=2));
  [`\cv{f$_1$}`??] -> (count'=count+1) & (inform'=(count+1=2));
  [`\cv{f$_2$}`??] -> (count'=count+1) & (inform'=(count+1=2));

  [`\cv{u$_0$}`??] -> (count'=count-1) & (inform'=(count=2));
  [`\cv{u$_1$}`??] -> (count'=count-1) & (inform'=(count=2));
  [`\cv{u$_2$}`??] -> (count'=count-1) & (inform'=(count=2));

  [`\cv{f}`!!] inform & count >= 2 -> (inform'=false);
  [`\cv{u}`!!] inform & count < 2 -> (inform'=false);
endmodule
  \end{lstlisting}
Voting gates are modeled using a counter which counts how many inputs have
failed. This is done by listening to the corresponding fail signals at lines 5
to 7, and repair signals at lines 9 to 11. In these same lines we take into
account if we have just reached the K value (2 in our example) or if we have
just gone down this value, which are the circumstances under which to inform
the failure and repair respectively, which is finally done at lines 13 and 14.
Although an alternative modelling of these gates can be obtained by a combination
of OR and AND gates, one may want to reduce the complexity of the system
modelling by using this model, which also happens to be deterministic.

\section{}
\label{appdx:proofOfNonConfluent}
\begin{proof}[of Proposition \ref{prop:nonconfluent}.]
  Parallel composition does not introduce new non-confluent pair of actions and,
  moreover, it preserves the confluency of its
  components~\cite{DArgenioMonti18}.  Thus, we look at the components in
  isolation.  First notice that transitions in an \iosa\ module are defined
  symbolically. Each symbolic transition in a module describes, in fact, a set
  of \iosa\ transitions, which become concrete when the symbolic transition is
  evaluated on a state that satisfies the guard. Notice also that a state in a
  module is defined by the current values of its variables.  When analyzing that
  two urgent actions $a$ and $b$ are confluent in a module, for each symbolic
  transition $t_a$ and $t_b$ defined for those actions in that module, we look
  for a \emph{non-confluence witness}, i.e, a state that satisfies the guards of
  $t_a$ and $t_b$ and shows that $a$ and $b$ are not confluent (i.e., the pair
  does not satisfy Def.~\ref{def:confluence}).  Note that by only checking
  reachable states in the component, we are already overapproximating the
  reachable states in the composition.

  For this proof we only analyze the case of the AND gate. For other \rft\
  elements, the proof follows similarly. Let $v$ be a vertex in a \rft\ such
  that $\lab(v)=(\AND,2)$. We analyze $\tv{f}_1$ against $\tv{u}_1$ in
  $[\![(\AND,2)]\!]$ (Fig. \ref{fig:pand}) and show that they are not confluent.
  Take state $s$ defined by \code{count=1}, \code{informf=false} and
  \code{informu=false}, which can be easily checked to be reachable. There, we
  find that it enables symbolic transitions at lines $6$ (with label
  $\tv{f}_1$) and $14$ (with label $\tv{u}_1$).  On the one hand, transition
  at line $6$ moves to the state where \code{count=2}, \code{informf=true} and
  \code{informu=false} is reached.  At this point action $\tv{u}_1$ can only
  be performed through transition at line $13$, which yields state $s'$ defined
  by \code{count=1}, \code{informf=true} and \code{informu=true}.
  On the other hand, transition at line $14$ moves to the state where
  \code{count=0}, \code{informf=false} and \code{informu=false}.  This state
  only enables $\tv{f}_1$ at line $7$, which yields state $s''$ defined by
  \code{count=1}, \code{informf=false} and \code{informu=false}.  Since $s'$ and
  $s''$ are two different states, we have proved that $\tv{f}_1$ and
  $\tv{u}_1$ are not confluent.  Similarly, we can show that the pairs
  $(\tv{f},\tv{u}_i$) and $(\tv{u},\tv{f}_i)$, for $i=1,2$, are not
  confluent.

  All other pairs are confluent. Take for instance transitions at lines $7$ and
  $10$ which are defined for actions \tv{f$_1$} and \tv{f$_2$} respectively,
  and the state $s$ defined by \code{count=0}, \code{informf=false} and
  \code{informu=false}.  On the one hand, line $7$ leads to the state where
  \code{count=1}, \code{informf=false} and \code{informu=false} which in turns
  enables \tv{f$_2$} only at line $9$ yielding state $s'$ defined by
  \code{count=2}, \code{informf=true} and \code{informu=false}. On the other
  hand, line $10$ at state $s$ moves to the state where \code{count=1},
  \code{informf=false} and \code{informu=false} which only enables \tv{f$_1$}
  at line line $6$ yielding the same state $s'$.  The proof follows similarly
  from any other reachable state enabling \tv{f$_1$} and \tv{f$_2$} showing,
  thus, that \tv{f$_1$} and \tv{f$_2$} are confluent.  In some other cases
  the proof of confluence follows from the fact that the pair of actions are
  never enabled simultaneously, as it is the case, e.g., of \tv{f} and
  \tv{u} (notice that the guards enabling each one of them are mutually
  exclusive). \qed
\end{proof}

\section{The Spare Gate model}
\label{app:spare}

\paragraph{The Spare basic element (SBE).} For a \emph{SBE}
element $e\in\E$, its semantics is a function
$[\![e]\!]:\actions^{7+5*n}\rightarrow\iosa$, where
$[\![ (\SP,n,\mu,\nu,\gamma) ]\!](\tv{fl},\tv{up},\tv{f},\tv{u},\tv{r},\tv{e},\tv{d},\tv{rq}_0,\tv{asg}_0,$\\$
\tv{rel}_0,\tv{acc}_0,\tv{rj}_0,...,\tv{rq}_{n-1}, \tv{asg}_{n-1},\tv{rel}_{n-1},\tv{acc}_{n-1},\tv{rj}_{n-1})$
results in the following pair of \iosa\ modules:

\begin{lstlisting}
module SBE
  fc, dfc, rc : clock;
  inform : [0..2] init 0;
  active : bool init false;
  broken : [0..2] init 0;

  [`\cv{e}`??] !active -> (active'=true) & (fc'=$\mu$);
  [`\cv{d}`??] active -> (active'=false) & (dfc'=$\nu$);

  [`\cv{fl}`!] active & broken=0 @ fc   -> (inform'=1) & (broken'=1);
  [`\cv{fl}`!] !active & broken=0 @ dfc -> (inform'=1) & (broken'=1);
  [`\cv{r}`??]                          -> (broken'=2) & (rc'=$\gamma$);
  [`\cv{up}`!] active & broken=2 @ rc   -> (inform'=2) & (broken'=0) & (fc'=$\mu$);
  [`\cv{up}`!] !active & broken=2 @ rc  -> (inform'=2) & (broken'=0) & (dfc'=$\mu$);

  [`\cv{f}`!!] inform=1 -> (inform'=0);
  [`\cv{u}`!!] inform=2 -> (inform'=0);
endmodule
\end{lstlisting}

\begin{lstlisting}
module MUX
  queue[n]: [0..3] init 0; % idle, requesting, reject, using
  avail: bool init true;
  broken: bool init false;
  enable: [0..2] init 0;

  [`\cv{fl}`?] -> (broken'=true);
  [`\cv{up}`?] -> (broken'=false);

  [`\cv{e}`!!] enable=1 -> (enable'=0);
  [`\cv{d}`!!] enable=2 -> (enable'=0);

  [`\cv{rq$_0$}`??]  queue[0]=0 & (broken | !avail) -> (queue[0]'=2);
  [`\cv{rq$_0$}`??]  queue[0]=0 & !broken & avail   -> (queue[0]'=1);
  [`\cv{asg$_0$}`!!] queue[0]=1 & !broken & avail   -> (queue[0]'=3) & (avail'=false);
  [`\cv{rj$_0$}`!!]  queue[0]=2 -> (queue[0]'=1);
  [`\cv{rel$_0$}`??] queue[0]=3 -> (queue[0]'=0) & (avail'=true) & (enable'=2);
  [`\cv{acc$_0$}`??] -> (enable'=1);
  ...
  [`\cv{rq$_{n-1}$}`??]  queue[n-1]=0 & (broken | !avail) -> (queue[n-1]'=2);
  [`\cv{rq$_{n-1}$}`??]  queue[n-1]=0 & !broken & avail   -> (queue[n-1]'=1);
  [`\cv{asg$_{n-1}$}`!!] queue[n-1]=1 & queue[n-2]=0 & ... & queue[0]=0 & !broken & avail
                  -> (queue[n-1]'=3) & (avail'=false);
  [`\cv{rj$_{n-1}$}`!!]  queue[n-1]=2 -> (queue[n-1]'=1);
  [`\cv{rel$_{n-1}$}`??] queue[n-1]=3 -> (queue[n-1]'=0) & (avail'=true) & (enable'=2);
  [`\cv{acc$_{n-1}$}`??] -> (enable'=1);

endmodule
\end{lstlisting}
The model for a Spare basic element consists in two \iosa\ modules. One of them
presents the behaviour of a basic element which can be enabled and disabled, and
an other module, the multiplexer, which presents the means to manage the
sharing of the SBE between the interested Spare Gates. In this case, we have
decided to model the multiplexer with a priority policy, which prioritizes
lower index input spare gates to higher indexed ones (notice assignment
transitions at line 15 and 22 of the multiplexer module.) Other kinds of
policies can be defined as for repair box gates. In the model, actions
\tv{rq$_i$} indicate that the spare gate input $i$ is requesting the spare.
\tv{acc$_i$} indicates that input $i$ accepts the spare that has previously
been assigned to it through action \tv{asg$_i$}. On the other hand action
\tv{rj$_i$} indicates that it rejects it. Action \tv{rel$_i$} indicates
that input $i$ is releasing the spare that has previously been assigned to it.
Finally actions \tv{e}  and \tv{d} enable and disable the spare basic element
when needed.

\paragraph{The Spare Gate (SG).}
For a spare gate element $e\in\E$ with priority policy, its semantics is
a function $[\![e]\!]:\actions^{4+7*n} \rightarrow\iosa$, where
$[\![ (\SG,n) ]\!](\tv{f},\tv{u},\tv{fl}_0,\tv{up}_0,\tv{fl}_1,\tv{up}_1,$
$\tv{rq}_1,\tv{asg}_1,\tv{acc}_1,$$\tv{rj}_1,\tv{rel}_1...,\tv{fl}_n,\tv{up}_n,\tv{rq}_n,\tv{asg}_n,\tv{acc}_n,\tv{rj}_n,\tv{rel}_n)$
results in the following \iosa:

\begin{lstlisting}
module SPAREGATE
  state: [0..4] init 0; // on main, request, wait, on spare, broken
  inform: [0..2] init 0;
  release: [-n..n] init 0;
  idx: [1..n] init 1;

  [`\cv{fl${_0}$}`?] state=0 -> (state'=1) & (idx'=1);
  [`\cv{up${_0}$}`?] state=4 -> (state'=0) & (inform'=2);
  [`\cv{up${_0}$}`?] state=3 & idx=1 -> (state'=0) & (idx'=1) & (release'=1);
  ...
  [`\cv{up${_0}$}`?] state=3 & idx=n -> (state'=0) & (idx'=1) & (release'=n);

  [`\cv{fl${_1}$}`?] state=3 & idx=1 -> (release'=1);
  ...
  [`\cv{fl${_n}$}`?] state=3 & idx=n -> (release'=n);

  [`\cv{rq${_1}$}`!!] state=1 & idx=1 -> (state'=2);
  ...
  [`\cv{rq${_n}$}`!!] state=1 & idx=n -> (state'=2);

  [`\cv{asg${_1}$}`??] state=0 | state=1 | state=3 -> (release'=1);
  [`\cv{asg${_1}$}`??] state=2 & idx=1 -> (release'=-1) & (state'=3);
  [`\cv{asg${_1}$}`??] state=4 -> (release'=-1) &  (state'=3) & (idx'=1) & (inform'=2);
  ...
  [`\cv{asg${_n}$}`??] state=0 | state=1 | state=3 -> (release'=n);
  [`\cv{asg${_n}$}`??] state=2 & idx=n -> (release'=-n) & (state'=3);
  [`\cv{asg${_n}$}`??] state=4 -> (release'=-n) & (state'=3) & (idx'=n) & (inform'=2);

  [`\cv{rj${_1}$}`??] state=2 & idx=1 -> (idx'=2) & (state'=1);
  [`\cv{rj${_2}$}`??] state=2 & idx=2 -> (idx'=3) & (state'=1);
  ...
  [`\cv{rj${_n}$}`??] state=2 & idx=n -> (state'=4) & (idx'=1) & (inform'=1);

  [`\cv{rel${_1}$}`!!] release=1 & !(state=3 & idx=1) -> (release'= 0);
  [`\cv{rel${_1}$}`!!] release=1 & state=3 & idx=1 -> (release'= 0) & (state'=1) & (idx'=1);
  ...

  [`\cv{rel${_n}$}`!!] release=n & !(state=3 & idx=n) -> (release'=0);
  [`\cv{rel${_n}$}`!!] release=n & state=3 & idx=n -> (release'= 0) & (state'=1) & (idx'=1);

  [`\cv{acc${_1}$}`!!] release=-1 -> (release'= 0);
  ...
  [`\cv{acc${_n}$}`!!] release=-n -> (release'=0);

  [`\cv{f}`!!] inform = 1 -> (inform'=0);
  [`\cv{u}`!!] inform = 2 -> (inform'=0);
endmodule
\end{lstlisting}
The Spare Gate model is using a priority policy over the available Spare BEs.
This means that when looking for a Spare BE, it will start asking for it to the
lower index inputs and go on with higher index until obtaining a replacement.
Other policies can be defined into the spare gate too, just as with the
multiplexer and the repair box. In the SG model, a variable \code{state}
distinguishes from when the SG is working with its main BE, requesting a SBE,
waiting for a response from its inputs, working on a SBE or broken. A vector
named \code{release} indicates for each SBE input $i$ when the SG has to release
(value $i$) or accept (value $-i$) the assignment of that SBE. A variable
\code{idx} indicates which of the inputs to request next. At line $7$ the SG
defines the transition which starts with the SBE acquiring protocol whenever the
main BE fails. The following transitions up to line $15$ are there to release
the acquired SBEs whenever they fail or the main BE is repaired. Transitions
from lines 17 to 19 are there to request for each available SBE. After doing so,
we need to wait for a response from the corresponding multiplexer
(\code{state'=2}). The request can be rejected (lines 29 to 32), and we proceed
by asking for the next SBE by setting \code{idx} to the corresponding value if
there is one, or by failing in case none of the SBE where available
(\code{state'=4} at line 32). A SBE can be assigned to us when not needed
anymore (lines $21$ and $25$), or when we where expecting it in order to avoid
failing (lines $22$ and $26$), or when we had already failed and thus we get
repaired by using it (lines $23$ and $27$). I may want to release a SBE when
it is assigned to me and I do not need it (lines $34$ and $38$) or when it fails
while I am using it (lines $35$ and $39$). Finally we accept assigned SBEs at
lines $41$ to $43$ and we signal failure at line $45$ and repair at line $46$.
To further understand the meaning and intuition of each transition we refer the
reader to the SBE description which heavily synchronizes its transitions with
the SG  model.

%============================================================================

\end{document}